\documentclass[ejsv2]{imsart}

\RequirePackage[numbers]{natbib}

\RequirePackage[colorlinks,citecolor=blue,urlcolor=blue]{hyperref}
\RequirePackage{graphicx}
\usepackage{graphicx}
\usepackage{multirow}

\usepackage{mathrsfs}
\usepackage[title]{appendix}
\usepackage{xcolor}
\usepackage{textcomp}
\usepackage{manyfoot}
\usepackage{booktabs}

\usepackage{listings}

\usepackage{natbib,graphicx,epsfig,pdfpages,longtable,url,amsmath,bbm,amsthm, afterpage, stmaryrd}
\usepackage[ruled]{algorithm2e}
\usepackage{dsfont, amsmath}

\usepackage{hyperref}
\usepackage{pdfpages}

\startlocaldefs

\theoremstyle{plain}

\newtheorem{theorem}{Theorem}[section]

\theoremstyle{definition}

\theoremstyle{remark}
\newtheorem{remark}{Remark}

\endlocaldefs

\begin{document}
\begin{frontmatter}
\title{Efficiently Computed Marginal Likelihoods\\ for Multivariate Mixture Models}

\runtitle{Marginal Likelihoods for Multivariate Mixture Models}

\begin{aug}

\author[A]{\fnms{Martin}~\snm{Metodiev}\ead[label=e1]{martin.metodiev@doctorant.uca.fr}\orcid{0009-0000-9432-3756}},
\author[C]{\fnms{Marie}~\snm{Perrot-Dockès}\ead[label=e3]{marie.perrot-dockees@u-paris.fr}},
\author[A,D]{\fnms{Pierre}~\snm{Latouche}\ead[label=e4]{pierre.latouche@uca.fr}}
\and
\author[E]{\fnms{Adrian E.}~\snm{Raftery}\ead[label=e5]{raftery@uw.edu}}

\address[A]{Laboratoire de Mathématiques Blaise Pascal, CNRS,
Université Clermont Auvergne\printead[presep={,\ }]{e1,e4}}

\address[C]{CNRS, MAP5,
Université Paris Cité\printead[presep={,\ }]{e3}}

\address[D]{Institut universitaire de France}

\address[E]{Departments of Statistics and Sociology,
University of Washington\printead[presep={,\ }]{e5}}

\runauthor{M. Metodiev et al.}
\end{aug}

\begin{abstract}
We propose a new methodology for efficiently computing marginal likelihoods for univariate and multivariate mixture models. The estimator that this methodology proposes computes the marginal likelihood from Markov chain Monte Carlo (MCMC) samples, is consistent, asymptotically normal and of finite variance. In addition, it is invariant to label switching, does not require posterior samples from hidden allocation vectors, and is efficiently approximated, even for an arbitrarily high number of components. Its computational efficiency is based on an asymptotically optimal ordering of the parameter space, which can in turn be used to provide useful visualisations. We test it in simulation settings where the true marginal likelihood is available analytically. It performs well against state-of-the-art competitors, even in multivariate settings with a high number of components. We demonstrate its utility for inference and model selection on univariate and multivariate data sets.
\end{abstract}



\end{frontmatter}

\section{Introduction}\label{sec1}

Choosing the right number of components for a mixture model is a topic of much discussion \citep[see][for reviews]{Fr06-FiniteMixtureAndMarkovSwitchingModels,Ce_et_al19-model_selection_mixture_models}. The marginal likelihood, also known as the evidence or the normalising constant, is a Bayesian model choice criterion that can be used for this task and that is optimal from a Bayesian point of view. More precisely, if $G$ denotes the number of mixture components, $Z(G)$ the marginal likelihood, and $\Pi(G)$ the prior on $G$, choosing $G$ such that it maximises $Z(G)\Pi(G)$ minimises the model selection error rate on average over $\Pi$ \citep{Je61-marginallikelihood}.
 
Unfortunately, the marginal likelihood often does not possess an analytic expression. A variety of algorithms have been proposed for its computation, none of which is clearly seen as a gold standard \citep[see][for a review]{Ll_et_al23-review_marglikestims}. In addition, there is an almost complete lack in the literature of marginal likelihood estimation techniques for mixture models that are multivariate and contain a high number of components. The goal of this article is to present a marginal likelihood estimator that can be used in this setting and that satisfies the following three desiderata: it is precise, generic and efficiently approximated as the number of components increases.

\paragraph{\textbf{Marginal likelihood estimation}}

Available approaches for marginal likelihood estimation do not satisfy the three desiderata described before. These include bridge sampling \citep{MeWo96-bridgesampling} (when using a generic, non-symmetric proposal such as the bridge sampling implementation of \citet{gronau2020}) and Chib's estimator \citep{Ch95-chibs_estimator}, which have been observed to be imprecise when used to estimate the marginal likelihood of mixture models \citep{Ne99-symmetry_problem_chibs,Ce_et_al19-model_selection_mixture_models}. 

Many precise estimators are not generic: they either require sampling schemes different from a simple MCMC sample from the posterior \citep[e.g., the proposals of][]{Skilling2006,FrPe08-powerposteriors}, specific Markov chain Monte Carlo (MCMC) samplers that can obtain samples from hidden allocation vectors \citep[e.g., the proposals of][]{No07-MargLikEmptyComponents,Ru_et_al10-clustered_marglikestims,Pe_et_al14-importance_sampling_mixtures,Ha_et_al22-QuickMargLikEstimInMixtures} or randomly permuted MCMC samples \citep[e.g., the proposals of][]{Fruhwirth2004,Fr19-bridge_sampling_mixtures_02}. A truly generic estimator would not have to depend on these restrictions. Indeed, there are many popular samplers that do not provide samples from hidden allocation vectors or randomly permuted parameters, such as the Hamiltonian Monte Carlo (HMC) algorithm implemented in Stan \citep{stan2022}, the classical Metropolis-Hastings algorithm \citep{Metropolis1953,Hastings1970} and the tempered MCMC algorithm \citep{Ne96-temperedMCMC}. 

Precise and generic estimators have been proposed in \citet{Be_et_al03-quick_marglikestim_01,LeRo16-evidence_mixture_models,Ce_et_al19-model_selection_mixture_models}. Unfortunately, these estimators quickly become intractable as the number of mixture components increases. This may be because their computational complexity is related to the total number of permutations that can be applied to the labels of the parameters. Indeed, this number increases superexponentially, leading to extremely large values, even for standard datasets. For example, in \citet{Bo_et_al19-MBCbook2010} a dataset is analysed for which the true number of components is estimated to be $G=15$. The number of permutations is equal to $G!=15!$ in this case, which is over one trillion. Dealing with this superexponential increase may well be the central bottleneck to these marginal likelihood estimation procedures.

\paragraph{\textbf{Our proposal}}

Building on the ideas of \citet{Be_et_al03-quick_marglikestim_01} and \citet{Re20-newthames_variation}, we propose a new methodology for efficiently computing marginal likelihoods in univariate and multivariate mixture models with a high number of components, while still adhering to the three desiderata of precision, generality and computational efficiency. The estimator that this methodology proposes is \begin{itemize}
    \item precise, since it is consistent, asymptotically normal, and optimal in a certain sense, even for the multimodal posterior densities that typically appear in mixture models,
    
    \item generic, since it is compatible with any MCMC algorithm that provides a sample that converges in distribution to the posterior distribution of the parameters,
    
    \item efficiently approximated via a novel graph-theoretic algorithm, even when the number of mixture components is large.
\end{itemize}

\paragraph{\textbf{Outline of the paper}}

The paper is organised as follows. In Section \ref{sec: previous work relating to the THAMES}, we describe previous work on which our estimator, the truncated harmonic mean estimator (THAMES) for mixture models, is built. Then, we introduce the THAMES for mixture models in Section \ref{ssec: truncated harmonic mean estimation for mixture models} and explain how it can be computed efficiently in Section \ref{ssec: Calculating the THAMES for mixture models efficiently}, using the principle of overlap graphs to obtain an ordering constraint. In addition, we  propose a model choice criterion, the criterion of overlap (CO), and show how to implement an identity from \citet{No04-MixtureModelsDifferentSizeLink,No07-MargLikEmptyComponents} in Section \ref{ssec: The THAMES for high dimensional mixture models in the case of empty components} to better deal with the phenomenon of empty components. We compare the THAMES to state-of-the-art estimators in simulation studies, where the true marginal likelihood is known, in Section \ref{sec: simulations}, and on real datasets, in Section \ref{sec: real datasets}. We conclude with a discussion in Section \ref{sec: discussion}.

\section{Previous work relating to the THAMES \label{sec: previous work relating to the THAMES}}
Let $\theta$ denote the model parameter of dimension $R$ and $Y$ the data. The Bayesian framework is considered, where $\pi(\theta)$ is the prior probability density function (PDF) and $L(\theta)=p(Y|\theta)$ is the likelihood.

Let $\theta^{(1)},\dots,\theta^{(T)}$ be a MCMC sample that converges in distribution to the posterior distribution of the parameters. The methodology that we propose builds on a class of estimators that are evaluated on a truncated harmonic mean of a function of the MCMC sample. This class of estimators has the following form, where $V(A)$ denotes the volume of a set $A$:\begin{align}
    \hat{Z}^{-1}=\frac{1}{T/2}\sum_{\substack{t=T/2+1,\\\theta^{(t)}\in A}}^T\frac{1/V(A)}{\pi(\theta^{(t)})L(\theta^{(t)})}.
\end{align} Under mild assumptions on $A$, $\hat{Z}^{-1}$ converges to the reciprocal of the marginal likelihood \citep{Me_et_al24-thames}. However, the choice of $A$ is crucial, since it can heavily influence the variance of $\hat{Z}^{-1}$, and thus the speed of its convergence. \citet{Me_et_al24-thames} suggested setting\begin{align}
    A=E_{\hat{\theta},\hat{\Sigma},c}=\{\theta|(\theta-\hat{\theta})^\intercal\hat{\Sigma}^{-1}(\theta-\hat{\theta})<c^2\},
\end{align} an ellipsoid centred at an estimate of the posterior mean $\hat{\theta}$ and scaled by an estimate of the posterior covariance matrix $\hat{\Sigma}$. Notice that the posterior sample is split: $\hat{\theta}$ and $\hat{\Sigma}$ are computed using the first half (or portion, more generally), while the marginal likelihood estimator is computed from the second half. This choice of $A$ was justified by the assumption of asymptotic normality of the posterior distribution, which is ensured in many standard settings \citep[see, e.g.,][]{HeydeJohnstone1979,Ghosal2000,Shen2002,Miller2021}. However, in the context of mixture models, the posterior is highly multimodal, even as the size of the dataset increases, so asymptotic posterior normality does not hold. Another choice of $A$ is needed.

The following theorem shows that, under mild assumptions, an optimal choice of $A$ is given by an $\alpha$-highest-posterior-density (HPD) region. The proofs of this and any other theorem presented in this article are given in the supplementary material.\begin{theorem}\label{thm-hpdregion_optimal}
Assume that $\pi(\theta)L(\theta)$ is continuous on its support and that $\pi(\theta)L(\theta)$ is non-constant, in the sense that the volume of each level set $V(\{\theta|\pi(\theta)L(\theta)=q\})=0$ for all $q>0$. Let \begin{align}
    H_\alpha=\{\theta|\pi(\theta)L(\theta)>q_\alpha\},
\end{align} where $q_\alpha$ is the constant for which the posterior probability of $\{\theta\in H_\alpha\}$ is equal to $\alpha$. If the draws $\theta^{(1)},\dots,\theta^{(T)}\sim p(\theta|Y)$ are independent, there exists a constant $\alpha\in(0,1]$ such that $A=H_\alpha$ minimises the variance of $\hat{Z}^{-1}$ over all measurable sets $A$, conditional on $Y$.
\end{theorem}\begin{remark}The assumption that $\theta^{(1)},\dots,\theta^{(T)}$ is an i.i.d sample from the posterior distribution is rarely satisfied exactly, but should be satisfied approximately, for example if a sufficiently large burn-in period is used and the sample is sufficiently heavily thinned.\end{remark}

Unfortunately, the volume of $H_\alpha$ is needed for the computation of $\hat{Z}^{-1}$. We know of no simple way to compute this volume in general. \citet{RoWr09-thames_on_hpd_region} suggested using a convex hull which contains $H_\alpha$, but approximating the volume of a convex hull is difficult when the dimension of the parameter space is large. This is avoided by the proposal of \citet{Re20-newthames_variation}, who suggested to instead use a set that is sufficiently close to $H_\alpha$, but whose volume is easy to compute, by taking the intersection of the ellipsoid $E_{\hat{\theta},\hat{\Sigma},c}$ with $H_\alpha$: \begin{align*}
    B_{\hat{\theta},\hat{\Sigma},c,\alpha}&=E_{\hat{\theta},\hat{\Sigma},c}\cap \{\theta|\pi(\theta)L(\theta)>\hat{q}_\alpha\}\\&=\{\theta|(\theta-\hat{\theta})^\intercal\hat{\Sigma}^{-1}(\theta-\hat{\theta})<c^2,\pi(\theta)L(\theta)>\hat{q}_\alpha\},
\end{align*} where $\hat{q}_\alpha$ is the empirical $(1-\alpha)$-quantile of the sample of unnormalised posterior values. 
This proposal requires additional Monte Carlo simulations to compute the volume of $B_{\hat{\theta},\hat{\Sigma},c,\alpha}$. This is done by using a technique that was successfully employed by \citet{Si_et_al08-hpd_integral_approximation} in reciprocal importance sampling: an i.i.d sample $\theta^{(1)}_{\text{Unif}},\dots,\theta^{(N)}_{\text{Unif}}$ from the uniform distribution on $E_{\hat{\theta},\hat{\Sigma},c}$ is simulated to calculate \begin{align}\label{eq: Balpha_approx_volume}
    \hat{V}(B_{\hat{\theta},\hat{\Sigma},c,\alpha})=  V(E_{\hat{\theta},\hat{\Sigma},c})\cdot\frac{1}{N}\sum^N_{j=1}\mathds{1}_{\{\theta|\pi(\theta)L(\theta)>\hat{q}_\alpha\}}(\theta^{(j)}_{\text{Unif}}).
\end{align} \citet{Re20-newthames_variation} implemented this estimator and obtained good results in a variety of situations. In  particular, in the context of mixture models, it performed well, but only after adding additional MCMC simulations from random permutation sampling via a technique from \citet{Fr01-labelswitching_random_permutations}. This runs against our desideratum of generality.

\section{Methods}
\subsection{Truncated harmonic mean estimation for mixture models \label{ssec: truncated harmonic mean estimation for mixture models}}

The standard mixture model setting is considered, where the data $Y\in\mathcal{M}_{n\times d}(\mathbbm{R})$ consist of an i.i.d sample of $n$ different $d$-dimensional observations, $Y=(Y_1,\dots,Y_n)$, with a distribution that is a mixture of distributions from the same family,\begin{align}\label{eq: mix_distribution}
    Y_i|\xi,\tau\stackrel{\text{i.i.d}}\sim\sum^G_{g=1}\tau_gf(\cdot|\xi_{g}),\quad i=1,\dots,n,
\end{align} where\; $G>1,\xi=(\xi_1,\dots,\xi_G)$ is a matrix containing the mixture component parameters, $\tau_1,\dots,\tau_G>0,\sum^G_{g=1}\tau_g=1$ are the mixture proportions represented by the vector $\tau=(\tau_1,\dots,\tau_G)$, and $f(\cdot|\xi_{g})$ is a density characterised by $\xi_{g}$. We are interested in estimating the marginal likelihood $Z(G)$ for different values of the number of mixture components $G$. The full parameter is denoted by $\theta=(\xi_1,\dots,\xi_G,\tau_1,\dots\tau_{G-1})$.

If the prior distribution of $(\tau_g,\xi_g)$ does not depend on the index $g$, which is commonly the case and an assumption that we make from now on, the posterior distribution is symmetric. This results in the well-known label-switching problem, where some parts of the MCMC sample randomly switch their labels. Luckily, it is possible to simulate a sample whose labels are not switched by relabelling, e.g., with Stephens' algorithm \citep{St00-Stephenslabelswitching} or the ECR algorithm \citep{PaIl10-ECR_algo}. In the former, the sample is permuted such that it minimises the posterior expected loss under a class of loss functions, while the latter is based on permuting the sample such that it follows the law of a non-symmetric distribution. The two algorithms give similar results in our experience.

Let $\theta^{(1)\star},\dots,\theta^{(T)\star}$ denote the posterior sample after relabelling. To ensure that the marginal likelihood estimator computed on this sample is still consistent, a procedure suggested in \citet{Be_et_al03-quick_marglikestim_01} is used. This procedure averages over the permutations $P_1,\dots,P_{G!}$ of the the posterior sample, which results in \begin{align}\label{eq: sym_estim}
    \hat{Z}^{-1}_{\textup{S}}(G)=\frac{1}{G!}\sum_{o=1}^{G!}\frac{1}{T/2}\sum_{\substack{t=T/2+1,\\P_o(\theta^{(t)\star})\in A}}^T\frac{1/V(A)}{\pi(\theta^{(t)\star})L(\theta^{(t)\star})},
\end{align} This estimator is symmetric, in the sense that its value does not change when the values of the parameters change. This makes it invariant to the problem of label switching. This invariance ensures consistency and asymptotic normality of the estimator under standard conditions on the MCMC sample, due to the following theorem.\begin{theorem}\label{thm-thames_consistency}
Suppose that $\theta^{(1)},\dots,\theta^{(T)}$ is ergodic and converges in distribution, as $T\to\infty$, to the symmetric posterior $p(\theta|Y)$ of a mixture model. Let $\theta^{(1)\star},\dots,\theta^{(T)\star}$ denote a relabelling of $\theta^{(1)},\dots,\theta^{(T)}$. For any fixed choice of $A$ for which $\pi\cdot L$ is bounded away from 0, in the sense that a $q>0$ exists such that $\pi(\theta)\cdot L(\theta)\geq q$ for all $\theta\in A$, any symmetric estimator defined by Equation \eqref{eq: sym_estim} almost surely converges to the reciprocal marginal likelihood $Z^{-1}(G)$ as $T\to\infty$. If $\theta^{(1)},\dots,\theta^{(T)}$ is also ergodic of degree 2, then $\sqrt{T}(\hat{Z}^{-1}(G)-Z^{-1}(G))$ converges, as $T\to\infty$, in distribution to the normal distribution with mean equal to 0 and variance depending on the choice of $A$.
\end{theorem}\begin{remark}The previous theorem presumed that the MCMC sample was an i.i.d sample from the posterior distribution, an assumption that is typically not true exactly. This theorem does \textbf{not} make such an assumption. In particular, it ensures that the marginal likelihood estimator is consistent and asymptotically normal even when it is computed on a relabelled sample, that is not i.i.d. The only assumptions in this theorem related to the MCMC sample are standard assumptions such as ergodicity.\end{remark}

\paragraph{\textbf{Our proposal}}

We suggest combining the proposal of \citet{Re20-newthames_variation} with the symmetrisation procedure used by \citet{Be_et_al03-quick_marglikestim_01} to obtain the THAMES for mixture models, \begin{align}\label{eq: symmetric_thames_intro}
    \hat{Z}^{-1}_{\text{THAMES}}(G)=\frac{1}{G!}\sum_{o=1}^{G!}\frac{1}{T/2}\sum_{\substack{t=T/2+1,\\P_o(\theta^{(t)\star})\in B_{\hat{\theta},\hat{\Sigma},c,\alpha}}}^T\frac{1/\hat{V}(B_{\hat{\theta},\hat{\Sigma},c,\alpha})}{\pi(\theta^{(t)\star})L(\theta^{(t)\star})}.
\end{align} 

We compute the volume of $B_{\hat{\theta},\hat{\Sigma},c,\alpha}$ by setting $N=T$ in Equation \eqref{eq: Balpha_approx_volume}. Note that a naive computation of the full sum over the $G!$ permutations is not necessary, since most permutations will evaluate to 0. We elaborate on this in the next section.

The THAMES for mixture models (hereafter just THAMES) is symmetric, consistent and asymptotically normal as the size $T$ of the simulations increases, under standard assumptions on the MCMC sample.\footnote{This follows from Theorem \ref{thm-thames_consistency}: setting $A=B_{\hat{\theta},\hat{\Sigma},c,\alpha}$ implies that $\pi\cdot L$ is bounded away from 0 by definition, so the conditions of this theorem are fulfilled, conditional on the first half of the posterior sample that is used to define $B_{\hat{\theta},\hat{\Sigma},c,\alpha}$. In addition, the volume is estimated by a Monte Carlo sample, which means that the THAMES is equal to the estimator given in this theorem with $A=B_{\hat{\theta},\hat{\Sigma},c,\alpha}$, distorted by the multiplicative ratio $\frac{V(B_{\hat{\theta},\hat{\Sigma},c,\alpha})}{\hat{V}(B_{\hat{\theta},\hat{\Sigma},c,\alpha})}$. However, since the multiplicative ratio almost surely converges to one, consistency and asymptotic normality are still preserved due to Slutsky's theorem.} Only the relabelled posterior sample is required for the computation of the THAMES, with the first half being used to estimate the posterior mean and covariance matrix.

The THAMES possesses two tuning-parameters, $c\in(0,\infty)$ and $\alpha\in(0,1)$, and we can use asymptotic results to choose them in an optimal way. In fact, the following theorem implies that a value of $c$ that is asymptotically equivalent to $\sqrt{R+1}$ leads to a version of the THAMES that minimises a lower bound on the variance over all mixtures of truncated harmonic mean estimators, under some assumptions. The theorem shows optimality for a version of the THAMES that uses an ellipsoid as a truncation set. However, in practice we use an intersection between an ellipsoid and an $\alpha$-HPD region to ensure robustness of the THAMES. The value of $\alpha$ is chosen such that the truncation set of the THAMES excludes posterior tail behaviour that is very different from the behaviour of a multivariate Gaussian, which is the distribution that is used in the following theorem. This is done by minimising the Kolmogorov distance with respect to this distribution. For the exact procedure that we use to choose $\alpha$, see Supplement B.

\begin{theorem}\label{thm-mix_thames}
Consider a convex combination of truncated harmonic mean estimators, \begin{align}\label{eq: mix_thames}
    \hat{Z}^{-1}_{\text{Mix}}=\sum_{k=1}^K\omega_k\frac{1}{T/2}\sum_{\substack{t=T/2+1,\\\theta^{(t)}\in A_k}}^T\frac{1/V(A_k)}{\pi(\theta^{(t)})L(\theta^{(t)})},
\end{align} where $\omega_1,\dots,\omega_K$ are weights that sum to 1 and $A_1,\dots,A_K$ are truncation sets within the posterior support. Suppose that the posterior density $p(\theta|Y)$ is a Gaussian mixture, $p(\theta|Y)=\sum^K_{k=1}a_k\textup{MVN}_R(\theta;m_k,S_k),$ where $\text{MVN}_R(\theta;m_k,S_k)$ denotes the Gaussian law of dimension $R$ with mean $m_k$ and covariance matrix $S_k$. If the draws $\theta^{(1)},\dots,\theta^{(T)}$ from the posterior distribution are independent, an upper bound on the variance of $\hat{Z}^{-1}_{\text{Mix}}$, conditional on $Y$, is minimised by $A_k=E_{m_k,S_k,c_R}, \omega_k=a_k,$ where $c_R$ is asymptotically equivalent to $\sqrt{R+1}$, i.e.,  $\lim_{R\to\infty}\frac{c_R}{\sqrt{R+1}}=1$.
\end{theorem}\begin{remark}When choosing an ellipsoid as a truncation set, the weights in Equation \eqref{eq: mix_thames} are equal to $\omega_k=1/G!$ so long as the posterior density of the relabelled sample is unimodal, since this leads to $G!$ modes in the posterior due to label switching. This isn't necessarily the case, since the phenomenon of genuine multimodality, that is, multimodality in the posterior density of the relabelled sample, can be observed \citep{GrFr09-labelswitching_genuine_multimodality}. In such cases, where the number of genuine modes is larger than 1, it may be useful to obtain a more optimal version of the THAMES by averaging different versions of the THAMES, each centred in a different genuine mode. These different genuine modes can for example be determined by applying the R package mclust \citep{Lu_et_al23-mclust} to the relabelled posterior sample. One could also use the approach of \citet{Na_et_al25-new_thames_version} to determine the different ellipsoids over which one should average the different versions of the THAMES. However, the assumption of only $G!$ symmetric modes may very well be justified by the large sample behaviour of the posterior: given an identifiable (up to permutations) mixture model, the posterior itself can often be observed to converge to a mixture of $G!$ multivariate Gaussian distributions, as pointed out in \citet{Fr06-FiniteMixtureAndMarkovSwitchingModels}. This is exactly the condition under which this theorem is verified.\end{remark}

\begin{remark}
The assumptions of this theorem are that the MCMC draws independently follow the posterior distribution and that the posterior is equal to a mixture of multivariate normal distributions. The former assumption may well hold approximately, for example if a sufficiently large burn-in is used and the MCMC chain is sufficiently heavily thinned. The latter assumption also usually holds approximately, since any distribution can be approximated by a mixture of multivariate Gaussians, under few assumptions \citep{Ba10-approximate_by_Gaussian_mixtures}.
\end{remark}

The definition of the THAMES is given in Equation \eqref{eq: symmetric_thames_intro}. In the rest of the section we present a multistep methodology to efficiently approximate the THAMES. While this methodology — unlike Equation \eqref{eq: symmetric_thames_intro} — is not simple, it drastically reduces the computational complexity of the THAMES, far below the superexponential complexity of other symmetric estimators used in this setting. A full pseudo-code implementation of this approximation is provided in Supplement B.

\subsection{Calculating the THAMES for mixture models efficiently\label{ssec: Calculating the THAMES for mixture models efficiently}}

A naive computation of Equation \eqref{eq: symmetric_thames_intro} would require $G!$ computation steps, one for each permutation of the $G$ different labels. This rapidly becomes infeasible computationally as $G$ increases. However, in practice, many of these computations can be avoided due to the fact that the support $B_{\hat{\theta},\hat{\Sigma},c,\alpha}$ is contained within a very simple geometric object, the ellipsoid $E_{\hat{\theta},\hat{\Sigma},c}$. We can hence reduce computation time by selecting a restricted number of components and only summing over those permutations which respect the order of these components within the ellipsoid $E_{\hat{\theta},\hat{\Sigma},c}$. The selection of the components is done via graph-theoretic tools, while the permutations are determined via quadratic discriminant analysis \citep[QDA; see for example][]{GhCr19-qda_optimal_if_normal}.

In the next three subsections, we propose a way to efficiently approximate the THAMES. First, we apply the principle of $``$overlap graphs$"$ to choose a maximum set of components that do not overlap. Then, we use this set to construct an ordering constraint via QDA. We use this ordering constraint to efficiently approximate the THAMES by limiting the number of permutations required in the computation. We emphasise that overlapping components are determined in the parameter space and not in the observation space, since the objective is to reduce the computation cost of Equation \eqref{eq: symmetric_thames_intro}, which is computed on the parameters from the posterior.

\subsubsection{Visualising the posterior distribution via overlap graphs}

Since  the radius $c=\sqrt{R+1}$ is taken as suggested by \citet{Me_et_al24-thames}, $E_{\hat{\theta},\hat{\Sigma},c}$ roughly estimates a 50\%-HPD-region on the relabelled posterior distribution. For more details, we refer the reader to \citet{Me_et_al24-thames}. In such a region, two components should not be overlapping if none of the components is superfluous. It can efficiently be verified (see Supplement B) if the hyperplane $\{\theta|\xi_{g_1}=\xi_{g_2}\}$ crosses the ellipsoid $E_{\hat{\theta},\hat{\Sigma},c}$. If this is the case, we say that there is overlap between the components $g_1$ and $g_2$. This could be an indicator of either $g_1$ or $g_2$ being superfluous since they could be the same component.  

\begin{figure*}[t]
    \centering
\includegraphics[scale=.4]{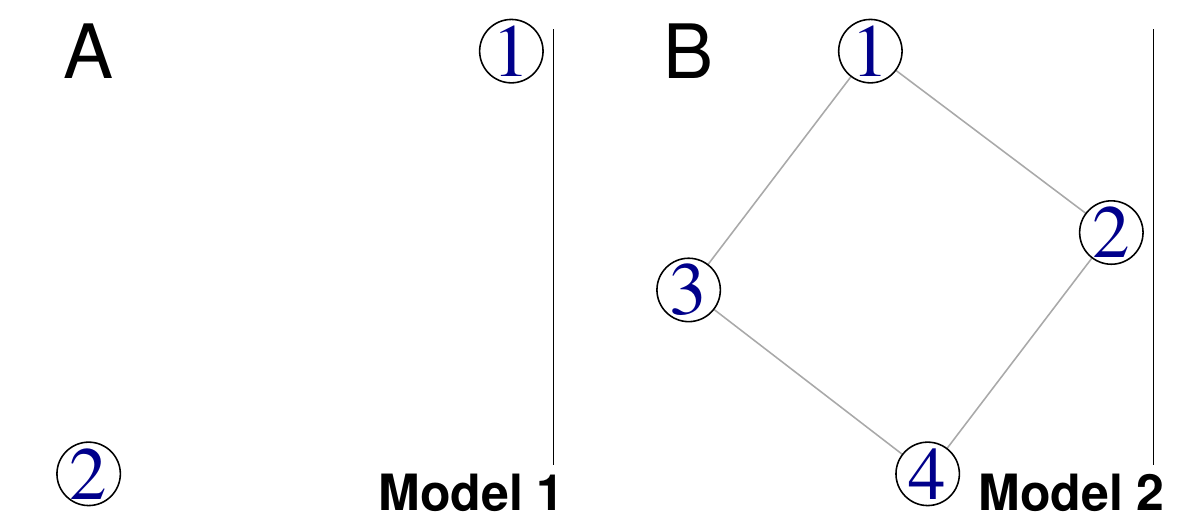}
    \caption{Example of two overlap graphs; an edge is added whenever there is overlap between 2 components.}
    \label{fig: separability graph}
\end{figure*}

 Figure \ref{fig: separability graph} shows two $``$overlap graphs$"$, where each node is a component and an edge is added between two nodes if there is overlap between the corresponding components. In Model 1, there is no overlap between all components, so there are no edges present. In Model 2, there is overlap between the pairs (1,2), (2,4), (4,3) and (3,1), so there are at most two non-overlapping components present. Please note that overlap graphs identify which mixture components have distinct parameters. In this sense, they achieve a different goal than most techniques that visualise the uncertainty of a cluster \citep[see for example][]{WaGh18-uncertainty_balls,Do_et_al24-dendrogram,LaLi24-visualization,BaWa25-understanding_uncertainty}.

Using graph-theoretic techniques on an overlap graph, we can find a set that approximates a maximum independent set, i.e., a maximum set of non-overlapping components. In Model 2 of Figure \ref{fig: separability graph} such a set is, for instance, given by components 1 and 4. This set can be determined via the build\_cover\_greedy function in the gor package \citep{gorpackage} and is denoted by $I(G)$. Note that the complexity of build\_cover\_greedy is quadratic with respect to $G$. As we are going to see, $I(G)$ will strongly help reducing the computational cost of Equation \eqref{eq: symmetric_thames_intro}. Interestingly, the size $|I(G)|$ can also be interpreted as an estimate of the number of distinct (i.e., non-overlapping) components. Thus, we define the criterion of overlap (CO) as \begin{align*}
    \text{CO}=\#\substack{\text{distinct}\\\text{components}} - \#\substack{\text{overlapping}\\\text{ components}}=|I(G)|-(G-|I(G)|).
\end{align*} 

The CO is related to the goal of determining the number of distinct components (i.e., clusters) instead of the number of true components. The difference between these two different goals was pointed out in \citet{Ba_et_al10-clustering_vs_densityestim}. While determining the number of components $G$ that maximises the marginal likelihood is useful in the latter case, choosing $G$ such that it maximises the CO is useful in the former. The set $I(G)$ that determines the CO is also going to be used to define an ordering constraint to estimate the marginal likelihood: nodes outside of $I(G)$ cannot be ordered together with nodes inside $I(G)$ and should thus not be used directly to define the constraint. The constraint will be constructed via QDA, and will in turn be used to quickly compute the THAMES.

\subsubsection{Applying QDA to the posterior sample}

Our aim is to find a classification of a component parameter $\xi_g$ such that the labels of two component parameters differ if, and only if, their classifications differ. This is done using discriminant analysis.

Focusing on the mixture component parameters in $\xi$, discriminant analysis can be performed on the relabelled posterior sample because the labels of the parameters are known by construction. Indeed, $\xi_1$ are the parameters of component $1$, $\xi_2$ of component $2$, and so on. Thus, a labelled sequence of size $G\cdot T$ from the relabelled posterior sample is given by \begin{align*}
    &\xi_1^{(1)\star},\dots,\xi_G^{(1)\star},\dots,\xi_1^{(T)\star},\dots,\xi_G^{(T)\star}.
\end{align*}

In the following, QDA is applied to this sequence. The idea is that QDA should be able to distinguish two samples from two  component parameters if, and only if, they are concentrated around two different values.

Let $\hat{\xi}_{1},\dots,\hat{\xi}_{G}$ and $\hat{S}_{1},\dots,\hat{S}_{G}$ denote the estimated component means and covariance matrices on the second half of the posterior sample, respectively. Let $\text{MVN}_u(\xi_{g'};\hat{\xi}_{g},\hat{S}_{g})$ denote the multivariate normal distribution of dimension $u$ with mean $\hat{\xi}_g$ and covariance matrix $\hat{S}_g$, evaluated at $\xi_{g'}$. Then, any value $\xi_{g'}\in\mathbb{R}^u$ (of which the true component $g'$ is unknown) will be assigned to \begin{align}
    &\hat{g}(\xi_{g'})=\text{argmax}_{g}\hat{w}_g(\xi_{g'}),\\&\hat{w}_g(\xi_{g'})=\frac{\frac{1}{G}\text{MVN}_u(\xi_{g'};\hat{\xi}_{g},\hat{S}_{g})}{\sum_{\tilde{g}=1}^G\frac{1}{G}\text{MVN}_u(\xi_{g'};\hat{\xi}_{\tilde{g}},\hat{S}_{\tilde{g}})}.\label{eq:qda} 
\end{align}Note that the cluster proportions are always estimated as $\frac{1}{G}$ by construction. Thus, $\xi^{(t)\star}_{g'}$ is assigned the wrong label if $\hat{g}(\xi^{(t)\star}_{g'})\neq g'$.

It should also be noted that, if $G$ is the true number of components, the relabelled posterior distribution should converge to the true value as the size of the data increases. Moreover, the condition needed to apply QDA is asymptotic normality, and posteriors are often asymptotically normal when converging to the truth \citep{HeydeJohnstone1979,Ghosal2000,Shen2002,Miller2021}. Since this condition is often met, we argue that an asymptotically optimal ordering constraint can be defined via QDA.

\subsubsection{Defining the ordering constraint}

 In the literature, ordering constraints have been employed to ensure identifiability \citep[see for example][]{RiGr97-datasets,StPh97-BayesianMethodsForNormalMixtures,LeDe00-MargLikstEstim_withconstraints,Ge07-simpleMCMCworks_whenconstraintorsymmetric}. However, the specific choice of the constraint can be of crucial importance \citep{Ce_et_al00-mixture_params_estims}. Results from the previous two sections will be applied to construct a constraint that is useful for our purposes, in the sense that the constraint will mostly hold within an approximate 50\%-HPD-region.

An approximate 50\%-HPD-region was used to define the set $I(G)$ such that a constraint could hold on this region for all of the components of $I(G)$. We now restrict the set of components over which the QDA allocation probabilities are computed. Thus, conditioning the QDA output on $I(G)$ gives the probability $\hat{w}_g(\xi_{g'}|I(G)) = \frac{\hat{w}_{g}(\xi_{g'})}{\sum_{\tilde{g}\in I(G)}\hat{w}_{\tilde{g}}(\xi_{g'})}$ for all $g\in I(G),g'\in\{1,\dots,G\}$.

Given normality, the assignment $\hat{g}$ and the uncertainty assessment $\hat{w}$ meet the optimality results provided by QDA \citep{GhCr19-qda_optimal_if_normal}. With this in mind, we combine this assignment $\hat{g}(\xi_{g'}|I(G))$ with its probability $\hat{w}_{\hat{g}(\xi_{g'}|I(G))}(\xi_{g'}|I(G))$ to obtain the summary measure $W(\xi_{g'})$:\begin{align*}
    W(\xi_{g'})=\hat{g}(\xi_{g'}|I(G))+1-\hat{w}_{\hat{g}(\xi_{g'}|I(G))}(\xi_{g'}|I(G)).
\end{align*}$W$ distinguishes the components in $I(G)$ by construction: if $\xi_{g'}^{(t)\star}$ is assigned to the component $\hat{g}(\xi_{g'}^{(t)\star}|I(G))$ by QDA, it is placed within $[\hat{g}(\xi_{g'}^{(t)\star}|I(G)),\hat{g}(\xi_{g'}^{(t)\star}|I(G))+1]$ by $W$. It is moved further away from $\hat{g}(\xi_{g'}^{(t)\star}|I(G))$ if its placement was assigned with low certainty.

\begin{figure}
    \centering
\begin{tabular}{c}\\
     \includegraphics[scale=.4]{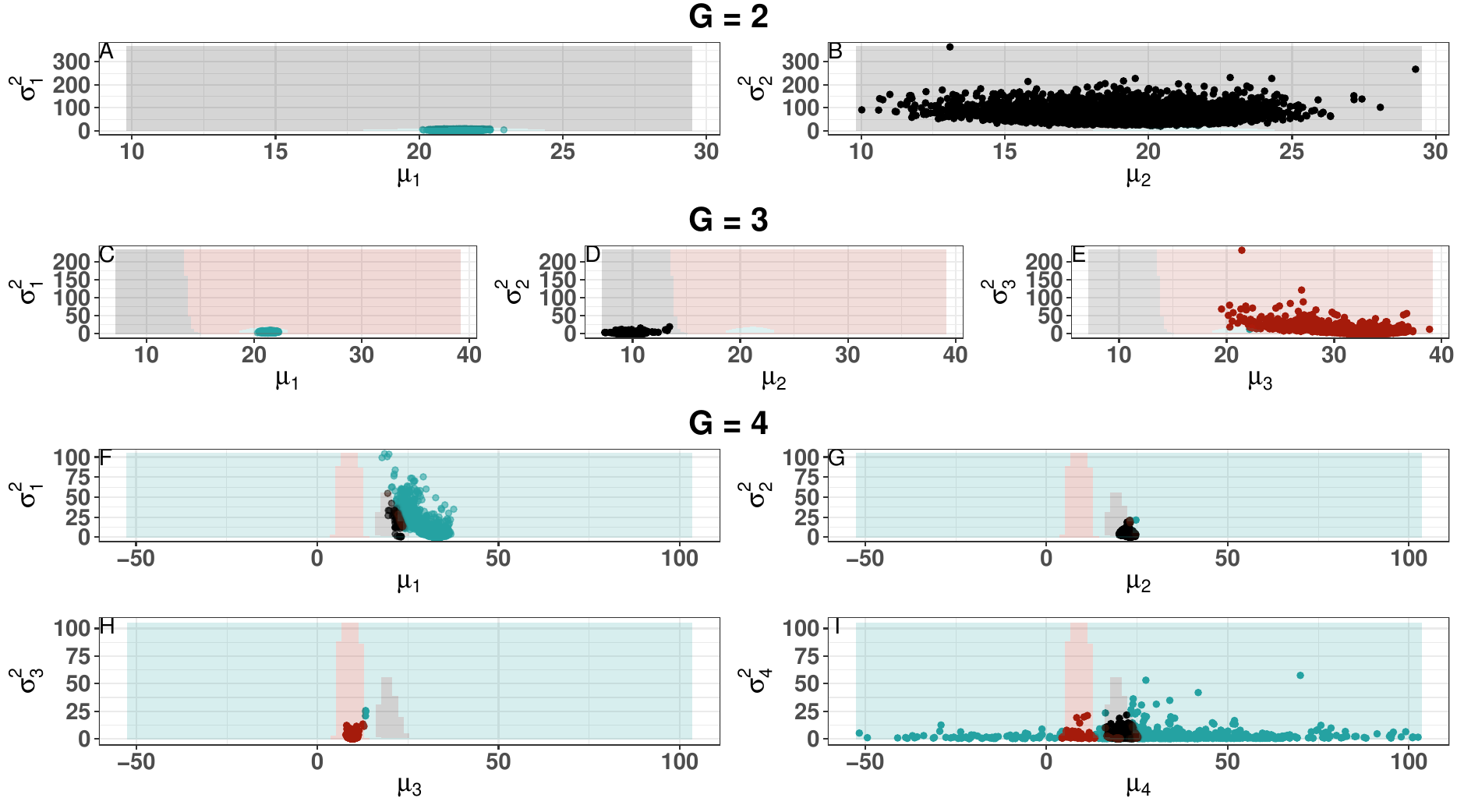}\\
      \includegraphics[scale=0.5]{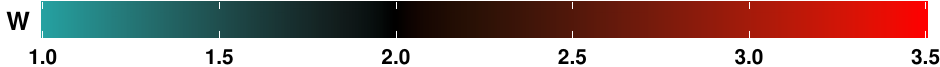}
\end{tabular}

    \caption{The values of $\mu,\sigma$ and $W$ for a relabelled sample $\xi^{(1)\star},\dots,\xi^{(T)\star}=(\mu^{(1)\star},\sigma^{(1)\star}),\dots,(\mu^{(T)\star},\sigma^{(T)\star})$ from the posterior of the galaxy dataset of size $T=10,000$ and different values of $G$; the acceptance regions of QDA are shown via a grid in the background}
    \label{fig:vis_galaxy}
\end{figure}

The role of the function $W$ can be explained via Figure \ref{fig:vis_galaxy}: it shows the different relabelled MCMC simulations of the component parameter $\xi=(\mu_1,\sigma_1),\dots,(\mu_G,\sigma_G)$ of the galaxy dataset (see Section \ref{sec: real datasets}) in conjunction with the values of $W$. Ordering by $\mu$ would not be enough in the case that $G=2$, since its values overlap, although the values of $\sigma$ do not. Ordering by $\sigma$ would not be enough when $G=3$ for the converse reason. However, the values of $W$ differ for most points whenever they are in a different component. The exception is component 4 in the case that $G=4$: it is superfluous, and was thus not included into the calculation of $W$.

$W$ is used to define the subspace \begin{align}\label{eq:identifiable_subset}
    \{\xi_1,\dots,\xi_G:W(\xi_1)\leq\dots\leq W(\xi_G)\}.
\end{align} This subspace will be used to subdivide the parameter space in $G!$ different subspaces $J_1,\dots,J_{G!}$ created by permuting the ordering constraint, e.g., $J_1=\{\xi_1,\dots,\xi_G:W(\xi_1)\leq W(\xi_2)\leq\dots\leq W(\xi_G)\},J_2=\{\xi_1,\dots,\xi_G:W(\xi_2)\leq W(\xi_1)\leq\dots\leq W(\xi_G)\}$ and so on. This subdivision will be used for the computation of the THAMES in Equation \eqref{eq: symmetric_thames_intro}.

\subsubsection{Efficiently computing the THAMES}

\begin{figure*}[t]
    \centering
    
    \includegraphics[scale=.25]{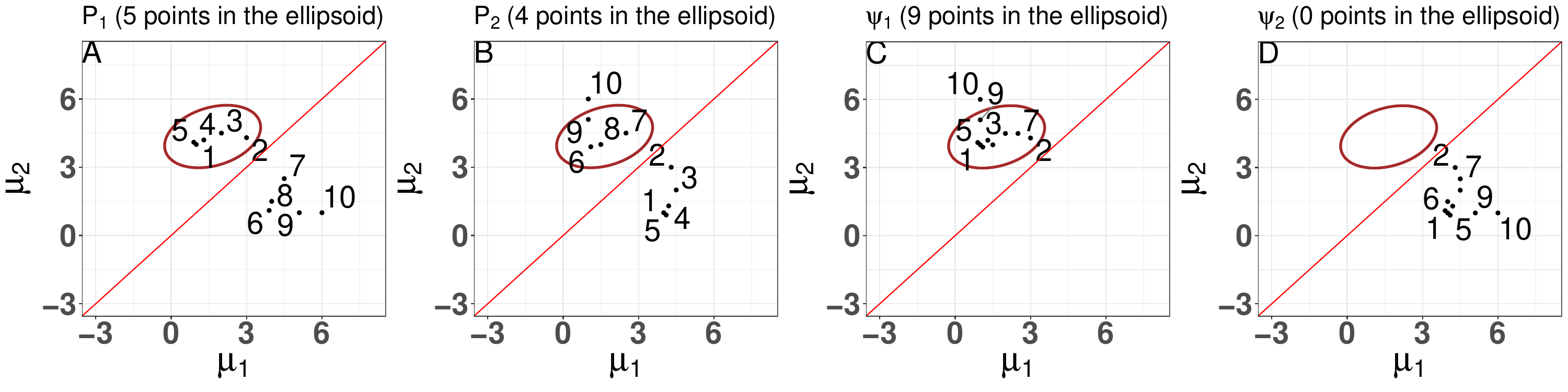}

    \caption{A toy example to illustrate the computation of the THAMES; the points that lie in the ellipsoid can be counted by applying $P_1,P_2$ or $\psi_1,\psi_2$.}
    \label{fig:dummy_example}
\end{figure*}

Let $\psi_1,\dots,\psi_{G!}$ denote the $G!$ different ordering algorithms that map $\theta$ to the different permutations $J_1,\dots,J_{G!}$ of the subspace defined in Equation \eqref{eq:identifiable_subset}, respectively, by permuting the labels of $\theta$. Note that these are different from $P_1,\dots,P_{G!}$, since $\psi_1,\dots,\psi_{G!}$ always enforce some version of Equation \eqref{eq:identifiable_subset} independent of the value of $(\xi_1,\dots,\xi_G)$. For example, $\psi_1$ may order $(\xi_1,\dots,\xi_G)$ to $(\xi_1^{\psi_1},\dots,\xi_G^{\psi_1})$ such that $W(\xi_1^{\psi_1})\leq W(\xi_{2}^{\psi_1})\leq \dots\leq W(\xi_{G}^{\psi_1})$ for all values of $(\xi_1,\dots,\xi_G)$, while $\psi_2$ would result into $W(\xi_{2}^{\psi_2})\leq W(\xi_{1}^{\psi_2})\leq\dots\leq W(\xi_{G}^{\psi_2})$ for all values of $(\xi_1,\dots,\xi_G)$.

The difference between $\psi_1,\dots,\psi_{G!}$ and $P_1,\dots,P_{G!}$ is shown in Figure \ref{fig:dummy_example} in the case that $G=2,T=10$ and that there are only two parameters $\theta^{(t)\star}=(\mu_1^{(t)\star},\mu_2^{(t)\star})$, one for each component, with $W$ replaced by the conceptually much simpler function $W'(\mu_g)=\mu_g$. Computing Equation \eqref{eq: symmetric_thames_intro} for $W'$ boils down to counting the number of times that $\psi_o(\theta^{(t)\star}),P_o(\theta^{(t)\star})$ falls in the ellipse for any $o=1,\dots,G!$. One can either apply $P_1$ (the identity) and $P_2$ (flipping the $x$- and $y$-axis), or $\psi_1$ (ordering the points such that $\mu_1^{(t)\star}\leq \mu_2^{(t)\star}$ for all $t=1,\dots,T$) and $\psi_2$ (ordering the points such that $\mu_2^{(t)\star}\leq\mu_1^{(t)\star}$ for all $t=1,\dots,T$). The outcome will be the same: 9 points are counted to be in the ellipse for some $\psi_o(\theta^{(t)\star}),P_o(\theta^{(t)\star})$, but computing $\psi_2(\theta^{(t)\star})$ could be avoided entirely since the ellipse does not cross the identity line.

We are going to determine which of $\psi_1,\dots,\psi_{G!}$ to evaluate through the following theorem. \begin{theorem}
Let $\psi_1,\dots,\psi_{G!}$ denote the $G!$ different ordering algorithms that map $\theta$ to different permutations of the ordering constraint \begin{align*}
    W(\xi_1)\leq\dots\leq W(\xi_G),
\end{align*} respectively, by permuting its labels. Let $\Delta$ denote the adjacency matrix defined by $\Delta_{g_1,g_2}=1$ if $W(\xi_{g_1})< W(\xi_{g_2})$ for all $\xi_{g_1},\xi_{g_2}$ that satisfy the ellipsoidal constraint on $E_{\hat{\theta},\hat{\Sigma},c}$, and let $\Omega$ denote any index set which includes the collection of all topological orderings on the graph defined by $\Delta$. The THAMES is equal to \begin{align}\label{eq:efficient_thames}
     \hat{Z}^{-1}_{\textup{THAMES}}(G)=\frac{1}{G!}\sum_{o\in \Omega}\frac{1}{T/2}\sum_{\substack{t=T/2+1\\\psi_o(\theta^{(t)\star})\in B_{\hat{\theta},\hat{\Sigma},c,\alpha}}}^T\frac{1/V(B_{\hat{\theta},\hat{\Sigma},c,\alpha})}{\pi(\theta^{(t)\star})L(\theta^{(t)\star})}.
\end{align}
\end{theorem}
\begin{remark}
A choice of $\Omega$ can be computed exactly, through a series of quadratic optimality problems that can be resolved in polynomial time with respect to $G$ and $T$ (see Supplement B). However, it is conceptually simpler and computationally more inexpensive to verify the inequalities $W(\xi_{g_1})< W(\xi_{g_2})$ approximately by re-using the Monte Carlo sample from Equation \eqref{eq: Balpha_approx_volume}. This approximation is used in all computations performed in this article. See Supplement B for details on these computations and the computation of $\Omega$, which can be done in polynomial time with respect to the size of $\Omega$. Note that any measurable function $W$ could have been used to construct $\Omega$. $W$ was chosen by QDA to target a function for which the inequalities $W(\xi_{g_1})<W(\xi_{g_2})$ are fulfilled for as many pairs $(g_1,g_2)$ as possible, thus reducing the size of $\Omega$.
\end{remark}

The parameters $c$ and $\alpha$ can always be chosen such that the set $\Omega$ is small enough and that the THAMES can be computed in a reasonable amount of time through Equation \eqref{eq:efficient_thames}: the smaller the chosen ellipsoid is, the more the overlap decreases. See Supplement B for more details. In practice we found that this is rarely necessary, since few of the different orderings need to be evaluated. This is likely due to the fact that the components overlap less and less as the size of the data $n$ increases.

\subsection{The THAMES for high dimensional mixture models in the case of empty components\label{ssec: The THAMES for high dimensional mixture models in the case of empty components}}

Aside from the computational challenges already discussed, employing mixture models in high dimensional settings can pose further difficulties. In particular, when the number of components is large, we run the risk of overfitting the data, resulting in empty components assigned to no data points. The posterior distribution of these empty component parameters is often highly diffuse, which makes it difficult to fit an ellipsoid around them, a technique that is vital for the precision of the THAMES. In this section, we will show how this problem can be sidestepped.

The component $g$ of a posterior sample entry $\theta^{(t)}$ is identified as possibly being empty if the product of conditional proportions \begin{align*}
    \prod_{i=1}^n(1-\hat{z}_{i,g}^{(t)})\text{, where }\hat{z}_{i,g}^{(t)}=\frac{\tau_g^{(t)}f(Y_i;\xi_g^{(t)})}{\sum_{\tilde{g}=1}^G\tau_{\tilde{g}}^{(t)}f(Y_i;\xi_{\tilde{g}}^{(t)})},
\end{align*} is not close to zero. This may be true for a large value of the total number of components $G$, but is going to be less and less likely as $G$ decreases. Thus, the empty component eventually disappears when simulating from the posterior of a sufficiently small model. In addition, when the probability of having an empty component is large, it is also efficiently estimated. Suppose that a Dirichlet prior with parameters $(e_1,\dots,e_G)>0$ is used for the proportions. We use an identity from \citet{No04-MixtureModelsDifferentSizeLink,No07-MargLikEmptyComponents} to recursively estimate the THAMES through the following theorem.

\begin{theorem}\label{thm-nobile_implication}
Suppose that a Dirichlet prior with parameters $e_1,\dots,e_G>0$ is used for the proportions and that $\hat{Z}(G-1)$ is a consistent estimator of $Z(G-1)$. If $\xi_1,\dots,\xi_G$ are a priori independent and their prior distribution does not depend on the total number of components, then\begin{align}\label{eq:multivariate_reduction_estim}
    \hat{Z}(G)=\frac{\Gamma\left(\sum^G_{g=1}e_g\right)\Gamma\left(n+\sum^{G-1}_{g=1}e_g\right)}{\Gamma\left(n+\sum^G_{g=1}e_g\right)\Gamma\left(\sum^{G-1}_{g=1}e_g\right)}\cdot \hat{Z}(G-1)\cdot\frac{1}{\hat{p}_0(G)},
\end{align} with \begin{align*}
    \hat{p}_0(G)=\frac{1}{G}\sum^G_{g=1}\frac{1}{T}\sum_{t=1}^T\prod_{i=1}^n(1-\hat{z}_{i,g}^{(t)}),\quad\hat{z}_{i,g}^{(t)}=\frac{\tau_g^{(t)}f(Y_i;\xi_g^{(t)})}{\sum_{\tilde{g}=1}^G\tau_{\tilde{g}}^{(t)}f(Y_i;\xi_{\tilde{g}}^{(t)})}
\end{align*}is a consistent estimator of $Z(G)$ as long as $\hat{p}_0(G)$ is a consistent estimator of its posterior mean.
\end{theorem}\begin{remark}Note that $\hat{p}_0(G)$ is an arithmetic mean, which implies that it is a consistent estimator of its posterior mean under standard conditions on the MCMC sample \citep{Ti94-MarkovChainsExploringPosteriors}. Thus, we can simply check if $\hat{p}_0(G)$ exceeds  some upper threshold (we chose the threshold $1/T$ in this article), and reduce the estimation procedure to the estimation of the marginal likelihood of the lower dimensional model if this is true.\end{remark}

It should be noted that this technique can only be used if $\xi_1,\dots,\xi_G$ are a priori independent and their prior distribution does not depend on the total number of components. We used this technique whenever this condition was fulfilled, which was the case in all multivariate Gaussian mixture models, and computed the THAMES without this technique otherwise.
\section{Experiments}
We now demonstrate the utility of the THAMES for mixture models on simulated and real datasets.
For conciseness, we limit the model descriptions to the definition of the mixture models and the definition of the datasets only. More detailed descriptions, such as the exact prior distributions used, are presented in Supplement C.

\subsection{Simulations\label{sec: simulations}}

We perform simulations for two settings in which the marginal likelihood can be either computed analytically or easily approximated with a high degree of accuracy. Formulas for these exact or approximate solutions are given in Supplement D. In the first, the exact expression of the marginal likelihood is computed in a brute-force manner, which can be done due to the low number of data points. In the second, the marginal likelihood is computed for a high number of extremely distinct components.

\subsubsection{Univariate Gaussian mixtures}

\paragraph{\textbf{Model}}

The mixture model is univariate Gaussian with $n=10$ data points,\begin{align}\label{eq: univ_gauss_mix}
    Y\in\mathbb{R}^{n},\quad  Y_{1},\dots,Y_{n}|\mu,\sigma,\tau\stackrel{\text{i.i.d}}\sim \sum_{g=1}^G\tau_g\mathcal{N}(\mu_g,\sigma^2_g),
\end{align} where only the parameters $\mu=(\mu_1,\dots,\mu_G)$ are unknown and $\tau=(\tau_1,\dots,\tau_G),\sigma=(\sigma_1,\dots,\sigma_G)$ are known.

\paragraph{\textbf{Settings}}

Simulations are conducted with the fitted and true $G=2$, but also when \textbf{underfitting}, where the true $G=3$ but the model is fitted for $G=2$, and \textbf{overfitting}, where the converse is true. Component overlap is varied via the parameter $\rho\in[0,1]$ defined in Supplement C. It measures how far the posterior means of the component parameters are away from each other. $\rho=0$ means very little overlap, while $\rho=1$ implies high overlap.

\paragraph{\textbf{Results}}
\begin{figure*}[t]
    \centering
\begin{tabular}{c}\\
     \includegraphics[scale=.6]{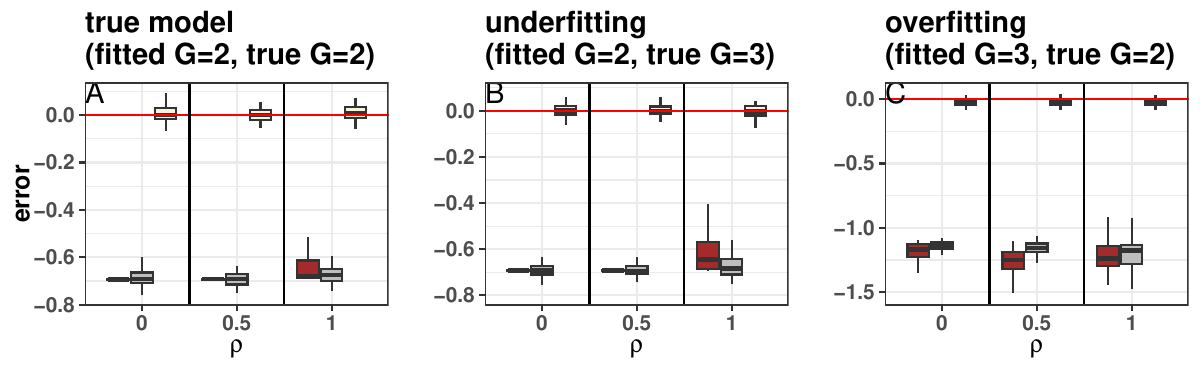}\\
      \includegraphics[scale=.8]{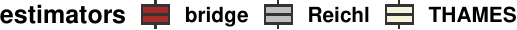}
\end{tabular}

    \caption{Boxplots of the error $\log(\hat{Z})-\log(Z)$ of bridge sampling, the estimator from \citet{Re20-newthames_variation}, and the THAMES for 50 different simulations of the data for increasing levels of overlap $\rho=0,1/2,1$ and 3 different scenarios corresponding to the true model, underfitting and  overfitting with $G=2,G=3$, respectively.}
    \label{fig:true_marglik_results}
\end{figure*}

The THAMES is compared to the bridge sampling \citep{MeWo96-bridgesampling} implementation of \citet{gronau2020} and the estimator from \citet{Re20-newthames_variation} (using the same tuning parameters as the THAMES) for 50 independently generated datasets, with $T=10,000$, a burn-in of 2,000 and using ECR for relabelling. Results are shown in Figure \ref{fig:true_marglik_results}. Boxplots of the THAMES overlapped with the truth in all cases. However, the other estimators are imprecise. 

\subsubsection{Multivariate Gaussian mixture models}

\paragraph{\textbf{Model}}

Let $Y_{1},\dots,Y_{n}$ be i.i.d values from a mixture of $G$ multivariate normal distributions, namely
\begin{align}\label{eq: gauss_mix_multi}
    Y\in\mathbb{R}^{n\times d},\quad  Y_{1},\dots,Y_{n}|\mu,\Sigma,\tau\stackrel{\text{i.i.d}}\sim \sum^G_{g=1}\tau_g\text{MVN}_d(\mu_g,\Sigma_g).
\end{align} The settings were each chosen such that they correspond to a real dataset.

\paragraph{\textbf{Settings}}

\begin{itemize}
    \item[(a)] Setting 1 (moderate $G$): the model is $d=6$-dimensional with $G=5$ components and $n=200$ data points.
    
    \item[(b)] Setting 2 (large $G$): the model is $d=5$-dimensional with $G=15$ components and $n=345$ data points. The number of free parameters $R$ is very large in this setting ($R=314$). As is typical in such high dimensional models \citep[see for example][]{Be_et_al97-BayesianGaussianMixtures}, it can be decreased by imposing additional constraints on the component covariance matrices. In this case, they are restricted to be diagonal matrices, resulting in much fewer free parameters ($R=164$).
\end{itemize}

\paragraph{\textbf{Results}}

\begin{figure*}[t]
    \centering
    \begin{tabular}{c}\\
    \includegraphics[scale=.6]{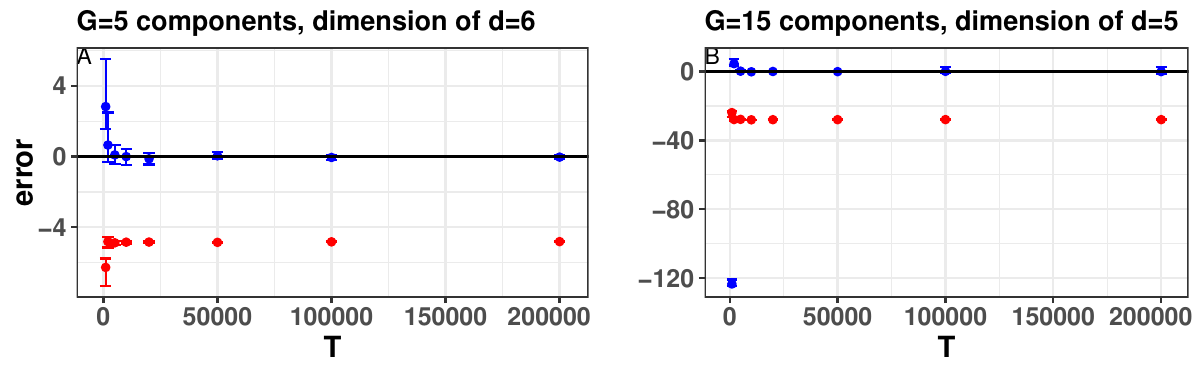}\\
    \includegraphics[scale=.4]{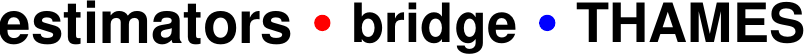}
    \end{tabular}

    \caption{The error $\log(\hat{Z})-\log(Z)$ of the THAMES and bridge sampling, as well as approximate 95\% confidence intervals, for different values of $T$; on the left: the unconstrained Gaussian mixture model with $G=5,d=6$; on the right: the constrained Gaussian mixture model with $G=15,d=5$.}
    \label{fig:true_marglik_gaussmulti_results}
\end{figure*}

The results are shown in Figure \ref{fig:true_marglik_gaussmulti_results}. The THAMES and bridge sampling were each computed on a MCMC sample with a burn-in of $2,000$, after relabelling via the ECR algorithm. Their evaluations are shown for different posterior sample sizes. The THAMES converges to the true value, while bridge sampling is off by a factor of $G!$. It should be noted that the strategy of simply shifting a non-symmetric marginal likelihood estimator by $G!$ does not work in general. It would work in this specific, artificial setting, only because all components are well separated. It does not work in others, as pointed out by \citet{Ne99-symmetry_problem_chibs}.

\subsection{Real datasets \label{sec: real datasets}}

The THAMES is computed for univariate and multivariate datasets, for which the number of components $G$ is unknown. The value of $G$ that maximises the THAMES can be used as an estimator for this number, since it minimises the model selection error rate on average, when considering a uniform prior on $G$.

\subsubsection{The univariate case}

The THAMES was evaluated for the datasets used in \citet{RiGr97-datasets}, namely the $``$acidity$"$, $``$enzyme$"$, and $``$galaxy$"$ dataset, available at \citet{Rpackage_multimode}. Results of other, related marginal likelihood estimators on the same datasets are available in \citet[Section 2.3.2]{Ce_et_al19-model_selection_mixture_models} and were thus used for comparison.

\paragraph{\textbf{Settings}} All datasets were fitted to a univariate Gaussian mixture model with hierarchical priors, varying $G$ between 2 and 6.

\paragraph{\textbf{Results}}

\begin{figure*}[t]
    \centering
    \begin{tabular}{c}\\
         \includegraphics[scale=.6]{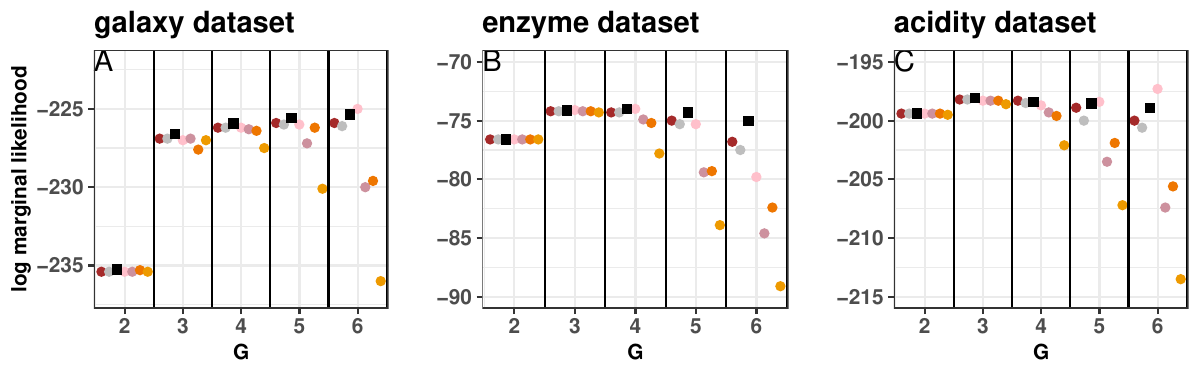}  \\
         \includegraphics[scale=1]{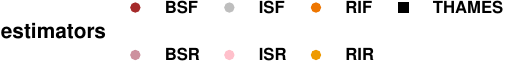}
    \end{tabular}
    \caption{Estimates from fully symmetric and random permutation bridge- (BSF, BSR), reciprocal- (RIF, RIR), importance sampling (ISF, ISR) and THAMES for the enzyme, galaxy and acidity datasets.}
    \label{fig:true_marglik_richardsongreen_results}
\end{figure*}

\begin{figure*}[t]
    \centering
\includegraphics[scale=.6]{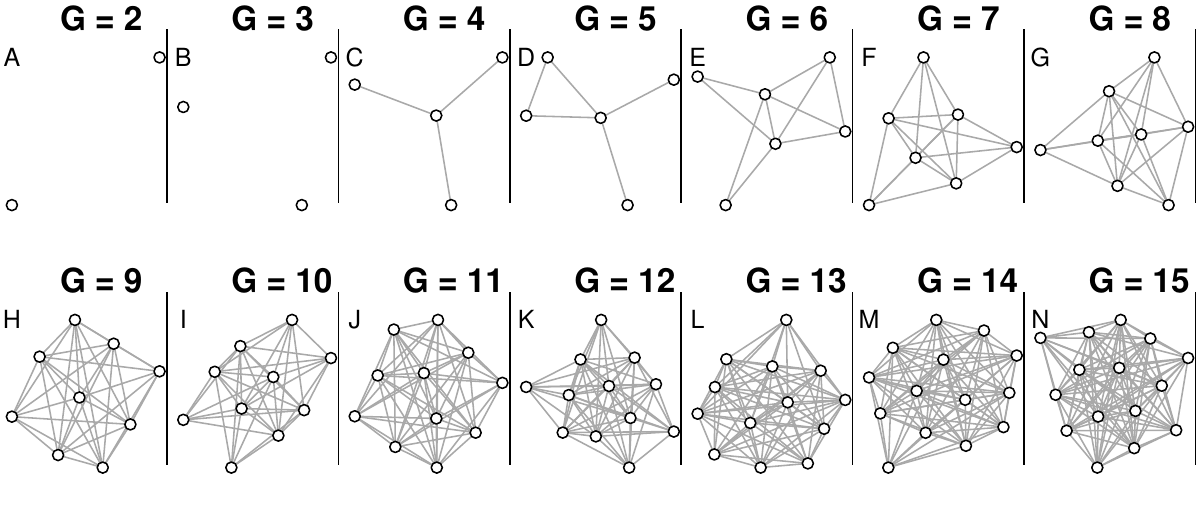}  
    \caption{Overlap graphs of the galaxy data set.}
    \label{fig:overlap_graphs_galaxies}
\end{figure*}

The THAMES estimates were compared to those provided by \citet{Ce_et_al19-model_selection_mixture_models} for other estimators under the same circumstances, namely $T=10,000$ and a burn-in of 2,000. The results are shown in Figure \ref{fig:true_marglik_richardsongreen_results}. The other estimators are fully symmetric and random permutation bridge sampling (BSF and BSR), importance sampling (ISF and ISR) and reciprocal importance sampling (RIF and RIR). Out of these, it was argued in \citet{Ce_et_al19-model_selection_mixture_models} that BSF and ISF performed well. Next to them, the median out of 20 independent THAMES estimates computed on different MCMC samples is shown. The THAMES is close to ISF and BSF in almost all settings, with slight differences occurring when $G=6$.

\begin{table}
    \centering
        \begin{tabular}{crrrrrrr}
         & G=2&\textbf{G=3}&G=4&G=5 &\textbf{G=6}&G=7&G=8\\\hline log-THAMES &-235.2 &\textbf{-226.7} &-226.0 &-225.6 &\textbf{-225.4} &-226.9 &-226.4\\ 
         95\% upper bound &-235.2 &\textbf{-226.6} &-225.8 &-225.2 &\textbf{-224.9} &-225.0 &-225.0\\
         95\% lower bound &-235.3 &\textbf{-226.7} &-226.2 &-226.0 &\textbf{-225.9} &-231.3 &-228.9\\
         CO&2   &\textbf{3}   &2   &1   &\textbf{0}  &-1  &-2\\
         $|\Omega|$&1    &1    &4   &20   &\textbf{60}  &840 &5040
         \\& G=9&G=10&G=11&G=12&G=13 &G=14&G=15\\\hline log-THAMES &-293.0 &-287.5 &-336.1 &-329.7 &-389.8 &-375.2 &-387.6\\ 
         95\% upper bound &-290.5 &-285.0 &-333.6 &-327.3 &-387.3 &-372.7 &-385.1\\
         95\% lower bound &-294.1 &-288.6 &-337.2 &-330.8 &-390.9 &-376.3 &-388.7\\
         CO&-5  &-4  &-7  &-6  &-9 &-10  &-9\\
         $|\Omega|$&42  &477    &2  &314   &10   & 6 &1170
         
    \end{tabular}
    \caption{The log-THAMES, as well as the upper and lower bound of its approximate 95\% confidence interval, the number of orderings $|\Omega|$ required for its computation, and the criterion of overlap (CO) for different values of $G$ for the galaxy dataset; the CO is maximised by $G=3$, the marginal likelihood by $G=6$. Numbers were rounded to the first decimal place. $|\Omega|$ is well below $G!$ in almost all cases, which is ensured by iteratively decreasing the size of the truncation set (see Supplement B for details on this procedure).
    \label{tab:marglikCO_galaxies}}
\end{table}

One limitation of the ISF and BSF estimators is that they become intractable as $G$ increases, for example if $G=15$. This is not the case for the THAMES. To illustrate this, it was computed from $G=2$ to 15 for the galaxy dataset with $T=100,000$ and a burn-in of $2,000$. The results are shown in Table \ref{tab:marglikCO_galaxies}. The THAMES peaked at $G=6$, while the criterion of overlap (CO) was maximised at $G=3$. The latter result is visualised with overlap graphs, in which the sharp increase of overlapping components at $G>3$ is apparent (Figure \ref{fig:overlap_graphs_galaxies}). The difference in the CO and the marginal likelihood in this setting is not surprising and can be interpreted as follows: the number of true components is estimated to be 6, while the number of distinguishable components is estimated to be 3. The $``$correct$"$ result depends on the interests of the practitioner \citep[see][]{Ba_et_al10-clustering_vs_densityestim}.

\subsubsection{The multivariate case}

Marginal likelihoods are computed for the Swiss banknote dataset, which contains 100 genuine and 100 counterfeit old-Swiss 1000-franc bank notes. They are also computed for the BUPA liver disorders dataset \citep{liver_disorders_60}, which contains 345 blood samples from different male individuals.

\paragraph{\textbf{Settings}}

The Swiss banknotes dataset consists of $n=200$ observations and $d=6$ variables. The BUPA liver disorders dataset consists of $n=345$ observations and $d=5$ variables. Multivariate Gaussian mixtures were fitted to both datasets. The covariance matrices of the BUPA liver disorders dataset were restricted to be diagonal matrices to deal with the high dimension of the parameters.

\paragraph{\textbf{Results}}

\begin{table}
    \centering
        \begin{tabular}{crrrr}
         & G=2&\textbf{G=3}&G=4&G=5 \\\hline log-THAMES& -926.79 &\textbf{-909.50} &-922.94 &-941.88\\
         95\% upper bound& -926.66 &\textbf{-909.23} &-922.67 &-941.61\\
         95\% lower bound& -926.90 &\textbf{-909.76} &-923.20 &-942.15\\
         
         CO&2&\textbf{3}&2&0
    \end{tabular}
    
    \caption{The log-THAMES, as well as the upper and lower bound of its approximate 95\% confidence interval, and criterion of overlap (CO) for different values of $G$ for the Swiss banknote dataset; both model choice criteria were maximised by $G=3$. Numbers were rounded to the second decimal place.}
    \label{tab:marglikCO_banknotes}
\end{table}

We set $T=10,000$ with a burn-in of 2,000 and apply ECR for relabelling. The THAMES was evaluated from $G=2$ to $G=5$ for the Swiss banknotes dataset (Table \ref{tab:marglikCO_banknotes}), with the maximum being obtained at $G=3$. This is consistent with the literature, where the estimated number of components for the Swiss banknotes dataset is often 2 or 3 \citep{CoCo23-numcluster_in_banknote_data}. 

\begin{table}
    \centering
        \begin{tabular}{cc}
        g & banknote type\\\hline
        \begin{tabular}{r}
        \\1\\2\\3
    \end{tabular}&\begin{tabular}{rr}
        genuine & counterfeit\\\hline
        1 & 16 \\
        99& 0\\0&84
    \end{tabular}
    \end{tabular}
    
    \caption{Confusion matrix for the Swiss banknote dataset for $G=3$, the value which maximises the marginal likelihood; g denotes the label of component g; components 1 and 3 seem to correspond to counterfeits, while component 2 seems to correspond to genuine banknotes.}
    \label{tab:confmat_banknotes}
\end{table}

The labels $``$genuine$"$ and $``$counterfeit$"$ are known for the banknote dataset. Table \ref{tab:confmat_banknotes} shows the confusion matrix with respect to these labels, where we clustered the data points for $G=3$, since this value was estimated to maximise the marginal likelihood and CO. All but one of the datapoints are assigned the $``$counterfeit$"$ label if belonging to component 1 or 3, the $``$genuine$"$ label if belonging to component 2.

\begin{table}
    \centering
        \begin{tabular}{crrrrrrr}
         & G=2&G=3&\textbf{G=4}&G=5 &G=6&G=7&G=8\\\hline log-THAMES &-6688.91& -6654.53& \textbf{-6647.91}& -6668.31& -6694.87& -6727.35& -6765.70\\
         $95\%$ upper bound &-6688.83& -6654.44& \textbf{-6647.75}& -6668.15& -6694.71& -6727.19& -6765.54\\
         $95\%$ lower bound &-6688.98& -6654.62& \textbf{-6648.06}& -6668.47& -6695.02& -6727.51& -6765.85\\
         CO&2&3&\textbf{4}&3&0&-1&-2\\& G=9&G=10&G=11&G=12&G=13 &G=14&G=15\\\hline log-THAMES &-6809.88& -6859.89& -6915.73& -6977.40& -7044.89& -7118.22& -7197.38\\
         $95\%$ upper bound &-6809.72& -6859.73& -6915.57& -6977.24& -7044.73& -7118.06& -7197.22\\
         $95\%$ lower bound &-6810.04& -6860.05& -6915.89& -6977.55& -7045.05& -7118.37& -7197.53\\
         CO&-3 &-4  &-5  &-6  &-7  &-8  &-9
    \end{tabular}
    
    \caption{The log-THAMES, as well as the upper and lower bound of its approximate 95\% confidence interval, and criterion of overlap (CO) for different values of $G$ for the BUPA liver disorders dataset; the CO and log-THAMES were both maximised by $G=4$. Numbers were rounded to the second decimal place.}
    \label{tab:marglikCO_liver}
\end{table}

For the BUPA liver disorders dataset, the THAMES was evaluated for $G$ between $2$ and $15$ (Table \ref{tab:marglikCO_liver}). It was maximised at $G=4$, with distinct levels of gamma-glutamyl transpeptidase for each component. For more details, we refer the reader to Supplement C.

\begin{table}
    \centering
        \begin{tabular}{cc}
        g & blood tests\\\hline
        \begin{tabular}{r}
        \\1\\2\\3\\4
    \end{tabular}&\begin{tabular}{rrrrr}
        mcv&  alkphos&      sgpt&     sgot&   \textbf{gammagt}\\\hline
        90.3 $\pm$ 0.5& 75.2 $\pm$ 2.2&  29.4 $\pm$ 1.1& 24.4 $\pm$ 0.6&  \textbf{37.8 $\pm$ 2.4}\\90.9 $\pm$ 0.8& 75.8 $\pm$ 2.4&  52.0 $\pm$ 3.2& 36.4 $\pm$ 1.9&  \textbf{91.3 $\pm$ 11.1}\\89.6 $\pm$ 0.3& 64.4 $\pm$ 1.2&  21.1 $\pm$ 0.6& 19.6 $\pm$ 0.4&  \textbf{17.7 $\pm$ 0.9}\\95.6 $\pm$ 2.1& 65.3 $\pm$ 9.2& 122.0 $\pm$ 26.1& 66.2 $\pm$ 8.8& \textbf{133.3 $\pm$ 22.9}
    \end{tabular}
    \end{tabular}
    
    \caption{Posterior mean estimate +/- the posterior standard deviation estimate of the different component means for the BUPA liver dataset for $G=4$, the value which maximised the marginal likelihood; the means correspond to the results of blood tests on the mean corpuscular volume (mcv), alkaline phosphotase (alkphos), alanine aminotransferase (sgpt), aspartate aminotransferase (sgot) and gamma-glutamyl transpeptidase (gammagt). g denotes the label of component g; the component means seem to strongly differ around the level of gammagt. Numbers were rounded to the first decimal place.}
    \label{tab:mupost_liver}
\end{table}

The component means estimated for the BUPA liver dataset seem to heavily differ around the variable $``$gammagt$"$, denoting the results of a blood test on gamma-glutamyl transpeptidase. This is indicated by the estimates of the posterior mean and standard deviation, computed on the relabelled posterior sample (Table \ref{tab:mupost_liver}).

\section{Discussion \label{sec: discussion}}

We presented the THAMES, a marginal likelihood estimator for mixture models. It is consistent, asymptotically normal, symmetric, and is computed on the relabelled MCMC parameters, without the need of simulations from the hidden allocation vectors. Its simplicity is exemplified by the fact that it only requires two inputs: the unnormalised log-posterior density (i.e., the logarithm of the prior times the likelihood) and the relabelled MCMC sample from the posterior. A function that takes exactly those two arguments is implemented in the R-package {\color{blue}\href{https://github.com/M-crypto645/thamesmix}{thamesmix}} \citep{thamesmix}.
    
The THAMES uses the symmetrisation method that was used in \citet{Be_et_al03-quick_marglikestim_01} and is tractable in the number of components, unlike other estimators that use this method. It should be noted that adjustments presented in \citet{Be_et_al03-quick_marglikestim_01} and \citet{LeRo16-evidence_mixture_models} already decrease the computation time of symmetrised estimators, but they do not lower it below $O(G!)$, which is the case for the THAMES. For example, we only needed to evaluate less than $10^{-6}$ percent of the permutations over which we sum in the case that $G=15$ for the galaxy dataset. However, for large $G$, higher overlap implies that the ellipsoid onto which the THAMES is defined will have to be shrunk to ensure that the THAMES can be computed in a reasonable amount of time, thus increasing the variance of the THAMES. 

The THAMES also uses an identity from \citet{No04-MixtureModelsDifferentSizeLink,No07-MargLikEmptyComponents} to deal with empty components in cases in which this identity is applicable. While this identity was always applicable for the multivariate mixture models that we used, it is not applicable in cases in which the prior distributions on the components are not independent, such as the cases in which a hierarchical prior is used. This was the case for the univariate datasets used in \citet{RiGr97-datasets}, where such a hierarchical prior was used on the variance parameters. In our simulations, the THAMES still performed well on these datasets, presumably because precisely such a hierarchical prior ensures that the probability of empty components occurring is low. Nevertheless, obtaining an identity that is similar to the identity from \citet{No04-MixtureModelsDifferentSizeLink,No07-MargLikEmptyComponents} but that applies to all mixture models would be interesting future work.

There are specific aspects of the THAMES — such as the fact that its support is a union of ellipsoids — that make it easy to compute, so we conjecture that our method of computation can be extended to more intricate, but similar marginal likelihood estimators \citep[e.g.,][]{Si_et_al08-hpd_integral_approximation,Wa_et_al17-pwk_estimators,Mc_et_al22-learntharmonicmean,Na_et_al25-new_thames_version}. It should also be noted that one of these estimators, the ECMLE, was introduced and compared to the THAMES in \citet{Na_et_al25-new_thames_version}, though not in the context of mixture models. Indeed, while the prior distribution is a mixture in one of the examples presented in that article, the likelihood is not. Thus, that comparison is not relevant to our paper, which focuses on marginal likelihood estimation in the context of mixture models, where the likelihood of the data, not the prior distribution, is a mixture.
    
Finally, we want to emphasise that our estimator can be applied to very general high dimensional and multivariate mixture models. This opens up a new avenue of numerous applications. One interesting application may be to further adapt the THAMES for the stochastic block model \citep{WaWo87-sbm,NoSn01-sbm}, in which only the complete data likelihood is available.

\section*{Author Contributions}
All four authors contributed to the paper.

\section*{Declarations}

\section*{Competing interests}
No competing interest is declared.

\section*{Data availability}
All code is made available at the following {\color{blue}\href{https://github.com/m-metodiev/thames-mixture-models}{link}}. The acidity, enzyme and galaxy datasets are available at \citet{Rpackage_multimode}, the Swiss banknote dataset is available at \citet{Lu_et_al23-mclust}, the BUPA dataset is available at \href{https://doi.org/10.24432/C54G67}{https://doi.org/10.24432/C54G67}.

\section*{Acknowledgments}
The authors would like to thank Qing Mai from Florida State University for the recommendation of QDA, as well as Alexander Reisach, Léo Hahn Lecler and Michael Schaller for fruitful discussions on partial orderings. In addition, they would like to thank Christian Robert and Bettina Grün for helpful comments. Finally, the authors would like to thank Nicholas J. Irons for his substantial help.

\section*{Funding}
Raftery's research was supported by NIH grant R01 HD-070936, the Fondation des Sciences Math\'{e}matiques de Paris (FSMP), Universit\'{e} Paris Cit\'{e} (UPC), and the Blumstein-Jordan professorship at the University of Washington. Latouche's research was supported by the Institut Universitaire de France (IUF).

\begin{supplement}
\stitle{Supplement A: proofs}
\sdescription{Includes the proofs to all theorems presented in this article.}
\end{supplement}
\begin{supplement}
\stitle{Supplement B: details on the algorithm}
\sdescription{Includes a pseudo-code algorithm of the THAMES and details on its implementation.}
\end{supplement}
\begin{supplement}
\stitle{Supplement C: details on the experiments}
\sdescription{Includes details on the simulations and the datasets.}
\end{supplement}
\begin{supplement}
\stitle{Supplement D: exact and approximate density computations}
\sdescription{Shows how we compute or approximate the marginal likelihood values used as baselines in the simulations.}
\end{supplement}

\bibliographystyle{chicago}
\bibliography{THAMES_mixtures_lib.bib}

\newcommand{\suppAtext}{Supplement A: proofs}
\newcommand{\textAone}{A.1 Proving that an optimal truncation set is given by a HPD region}
\newcommand{\textAtwo}{A.2 Proving that the THAMES is consistent and asymptotically normal}
\newcommand{\textAthree}{A.3 Proving that the variance of the THAMES minimises an upper bound}
\newcommand{\textAfour}{A.4 Proving that we can efficiently compute the THAMES}
\newcommand{\textAfive}{A.5 Proving that we can avoid dealing with empty components}

\newcommand{\suppBtext}{Supplement B: details on the algorithm}
\newcommand{\textBone}{B.1 Choosing the tuning-parameters}
\newcommand{\textBtwo}{B.2 Choosing the grid for $\alpha$}
\newcommand{\textBthree}{B.3 Numerically verifying that $\xi_{g_1}=\xi_{g_2}$ lies in $E_{\hat{\theta},\hat{\Sigma},c}$}
\newcommand{\textBfour}{B.4 Efficiently computing the THAMES}
\newcommand{\textBfive}{B.5 The full algorithm}

\newcommand{\suppCtext}{Supplement C: details on the experiments}
\newcommand{\textCone}{C.1 Simulations on univariate Gaussian mixtures}
\newcommand{\textCtwo}{C.2 Simulations on multivariate Gaussian mixtures}
\newcommand{\textCthree}{C.3 Univariate datasets (galaxy, enzyme, acidity)}
\newcommand{\textCfour}{C.4 Multivariate datasets (Swiss banknotes, BUPA liver)}

\newcommand{\suppDtext}{Supplement D: exact and approximate density computations}
\newcommand{\textDone}{D.1 Exact marginal likelihood for a small dataset}
\newcommand{\textDtwo}{D.2

Highly accurate 
marginal likelihood 
approximation 
for an extremely well separated dataset}
\noindent
\newpage
\begin{center}
{\large \textbf{Table of contents}}\\
\end{center}

\begin{flushleft}
{\large\textbf{\suppAtext}} \hfill {\color{blue}\hyperref[sec: suppA]{\pageref{sec: suppA}}} \\
{ \textAone} \hfill {\color{blue}\hyperref[ssec: Aone]{\pageref{ssec: Aone}}} \\
{\textAtwo} \hfill {\color{blue}\hyperref[ssec: Atwo]{\pageref{ssec: Atwo}}} \\
{\textAthree} \hfill {\color{blue}\hyperref[ssec: Athree]{\pageref{ssec: Athree}}} \\
{\textAfour} \hfill {\color{blue}\hyperref[ssec: Afour]{\pageref{ssec: Afour}}} \\
{ \textAfive} \hfill {\color{blue}\hyperref[ssec: Afive]{\pageref{ssec: Afive}}} \\
\textbf{\suppBtext} \hfill {\color{blue}\hyperref[sec: suppB]{\pageref{sec: suppB}}}\\ 
{\textBone} \hfill {\color{blue}\hyperref[ssec: Bone]{\pageref{ssec: Bone}}} \\
{\textBtwo} \hfill {\color{blue}\hyperref[ssec: Btwo]{\pageref{ssec: Btwo}}} \\
{\textBthree} \hfill {\color{blue}\hyperref[ssec: Bthree]{\pageref{ssec: Bthree}}} \\
{\textBfour} \hfill {\color{blue}\hyperref[ssec: Bfour]{\pageref{ssec: Bfour}}} \\
{\textBfive} \hfill {\color{blue}\hyperref[ssec: Bfive]{\pageref{ssec: Bfive}}} \\
\textbf{\suppCtext} \hfill {\color{blue}\hyperref[sec: suppC]{\pageref{sec: suppC}}} \\
{\textCone} \hfill {\color{blue}\hyperref[ssec: Cone]{\pageref{ssec: Cone}}} \\
{\textCtwo} \hfill {\color{blue}\hyperref[ssec: Ctwo]{\pageref{ssec: Ctwo}}} \\
{\textCthree} \hfill {\color{blue}\hyperref[ssec: Cthree]{\pageref{ssec: Cthree}}} \\
{\textCfour} \hfill {\color{blue}\hyperref[ssec: Cfour]{\pageref{ssec: Cfour}}} \\
\textbf{\suppDtext} \hfill {\color{blue}\hyperref[sec: suppD]{\pageref{sec: suppD}}} \\
{\textDone} \hfill {\color{blue}\hyperref[ssec: Done]{\pageref{ssec: Done}}} \\
{\textDtwo} \hfill {\color{blue}\hyperref[ssec: Dtwo]{\pageref{ssec: Dtwo}}}
\end{flushleft}
\newpage

\section*{\suppAtext \label{sec: suppA}}

Please note that the optimality results in Supplements {\color{blue}\hyperref[ssec: Aone]{A1}} and {\color{blue}\hyperref[ssec: Athree]{A3}} require additional assumptions on the distribution of the MCMC sample, as pointed out in the main document. However, the results on consistency, asymptotic normality and computational efficiency in Supplements {\color{blue}\hyperref[ssec: Atwo]{A2}} and {\color{blue}\hyperref[ssec: Afour]{A4}} make no such assumptions.

\subsection*{\textAone\label{ssec: Aone}}

The following theorem proves that the variance of the estimator \begin{align*}
    \hat{Z}^{-1}=\frac{1}{T/2}\sum_{\substack{t=T/2+1,\\\theta^{(t)}\in A}}^T\frac{1/V(A)}{\pi(\theta^{(t)})L(\theta^{(t)})},
\end{align*} is minimised by an $\alpha$-HPD region under conditions on which this region is well-defined, if the samples $\theta^{(1)},\dots,\theta^{(T)}$ independently follow the posterior distribution. The condition on the distribution of $\theta^{(1)},\dots,\theta^{(T)}$ is rarely satisfied exactly, but should be satisfied approximately, for example if a sufficiently large burn-in period is used and the sample is sufficiently heavily thinned.

\begin{theorem}[Theorem 2.1 in the main article]
Assume that $\pi(\theta)L(\theta)$ is continuous on its support and that $\pi(\theta)L(\theta)$ is non-constant, in the sense that the volume of each level set $V(\{\theta|\pi(\theta)L(\theta)=q\})=0$ for all $q>0$. Let \begin{align}
    H_\alpha=\{\theta|\pi(\theta)L(\theta)>q_\alpha\},
\end{align} where $q_\alpha$ is the constant for which the posterior probability of $\{\theta\in H_\alpha\}$ is equal to $\alpha$. If the draws $\theta^{(1)},\dots,\theta^{(T)}\sim p(\theta|Y)$ are independent, there exists a constant $\alpha\in(0,1]$ such that $A=H_\alpha$ minimises the variance of $\hat{Z}^{-1}$ over all measurable sets $A$, conditional on $Y$.
\end{theorem}

\begin{proof}
Let $A$ be an arbitrary set of non-zero, finite volume that is contained in the posterior support. Because of independence and unbiasedness, the variance of $\hat{Z}^{-1}$ is given by \begin{align*}
    \text{Var}[\hat{Z}^{-1}|Y]&=\text{Var}\left[\left.\frac{1}{T/2}\sum_{\substack{t=T/2+1,\\\theta^{(t)}\in A}}^T\frac{1/V(A)}{\pi(\theta^{(t)})L(\theta^{(t)})}\right|Y\right]\\&=\frac{1}{T/2}\text{Var}\left[\left.\frac{\mathds{1}_{A}(\theta^{(1)})/V(A)}{\pi(\theta^{(1)})L(\theta^{(1)})}\right|Y\right]\\&=\frac{1}{T/2}\text{E}\left[\left.\left(\frac{\mathds{1}_{A}(\theta^{(1)})/V(A)}{\pi(\theta^{(1)})L(\theta^{(1)})}\right)^2\right|Y\right]-\frac{Z^{-2}}{T/2}\\&=\frac{1}{T/2}\int_A\left(\frac{1/V(A)}{\pi(\theta)L(\theta)}\right)^2p(\theta|Y)\;d\theta-\frac{Z^{-2}}{T/2}\\&=\frac{Z^{-2}}{T/2}\int_{A} \frac{1/V(A)^2}{p(\theta|Y)}d\theta-\frac{Z^{-2}}{T/2},
\end{align*} where $\mathds{1}_A(\theta^{(1)})$ denotes the indicator function of $A$ evaluated at $\theta^{(1)}$. We are going to show that there exists a level $\alpha$ such that $\hat{Z}^{-1}$ has an equal- or lower variance under $H_\alpha$.

\textbf{Proving that there exists an $\alpha$ such that $H_{\alpha}$ is optimal}

Let $H_\alpha=\{\theta|p(\theta|Y)\geq q_\alpha\}$ denote an $\alpha$-HPD region. Since $V(H_\alpha)$ is continuous, we can choose $\alpha$ such that \begin{align*}
    V(H_\alpha)=V(A).
\end{align*} We have just shown that the variance of $\hat{Z}^{-1}$ is, up to translation, proportional to \begin{align}\label{eq: var_thames}
    \text{Var}[\hat{Z}^{-1}|Y]\propto\int_{A} \frac{1/V(A)^2}{p(\theta|Y)}d\theta.
\end{align} Comparing this to the same integral of $\hat{Z}^{-1}$ under $H_\alpha$ gives \begin{align*}
    \text{Var}[\hat{Z}^{-1}|Y]&\propto\frac{1}{V(A)^2}\int_{A} \frac{1}{p(\theta|Y)}d\theta\\&=\frac{1}{V(H_\alpha)^2}\int_{A\cap H_\alpha} \frac{1}{p(\theta|Y)}d\theta+\frac{1}{V(H_\alpha)^2}\int_{A\backslash H_\alpha} \frac{1}{p(\theta|Y)}d\theta\\&\geq \frac{1}{V(H_\alpha)^2}\left(\int_{A\cap H_\alpha} \frac{1}{p(\theta|Y)}d\theta+\frac{1}{q_\alpha}V(A\backslash H_\alpha)\right).
\end{align*} This is due to the fact that $A\backslash H_\alpha=\{\theta\in A|p(\theta|Y)< q_\alpha\}.$ Using this and the fact that $V(A\backslash H_\alpha)=V(H_\alpha \backslash A)$ one can similarly conclude \begin{align*}
   &\frac{1}{V(H_\alpha)^2}\left(\int_{A\cap H_\alpha} \frac{1}{p(\theta|Y)}d\theta+\frac{1}{q_\alpha}V(A\backslash H_\alpha)\right)=\\&\frac{1}{V(H_\alpha)^2}\left(\int_{A\cap H_\alpha} \frac{1}{p(\theta|Y)}d\theta+\frac{1}{q_\alpha}V(H_\alpha\backslash A)\right)\geq\\&\frac{1}{V(H_\alpha)^2}\int_{A\cap H_\alpha}\frac{1}{p(\theta|Y)}\;d\theta+\frac{1}{V(H_\alpha)^2}\int_{H_\alpha\backslash A} \frac{1}{p(\theta|Y)}\;d\theta=\\&\frac{1}{V(H_\alpha)^2}\int_{H_\alpha} \frac{1}{p(\theta|Y)}\;d\theta.
\end{align*} It follows that for any set $A$ with $V(A)\in(0,\infty)$ that is contained in the posterior support an $\alpha$ exists such that the variance of $\hat{Z}^{-1}$ is lower or equal when choosing $H_\alpha$ instead of $A$. This finishes the proof.

Please note that $H_\alpha$ is well-defined and $V(H_\alpha)$ is strictly increasing and continuous. This follows from the conditions given in the theorem. We could not find a direct proof in the literature (though the result is probably well-known), so we give one here.

\textbf{Proving that $H_\alpha$ is well-defined and $V(H_\alpha)$ is strictly increasing and continuous}

Let $q\in(0,\infty)$ be arbitrary. Let $P(\cdot|Y)$ denote the probability measure of the posterior distribution. It is well-known that any measure is continuous, in the sense that \begin{align*}
    \lim_{i\to\infty}P(A_i|Y)=P\left(\left.\bigcap_{i=1}^\infty A_i\right|Y\right)
\end{align*} for any decreasing series of sets $A_1\supseteq A_2\supseteq\dots$ \citep{Kl13-probabilitycourse}. Thus\begin{align*}
    &\lim_{\tilde{q}\uparrow q}P(\{\theta|\pi(\theta)L(\theta)>\tilde{q}\}|Y)=P\left(\left.\bigcap_{\tilde{q}\uparrow q}\{\theta|\pi(\theta)L(\theta)>\tilde{q}\}\right|Y\right)\\&=P(\{\theta|\forall \tilde{q}<q:\pi(\theta)L(\theta)>\tilde{q}\}|Y)=P(\{\theta|\pi(\theta)L(\theta)\geq q\}|Y)=P(\{\theta|\pi(\theta)L(\theta)> q\}|Y),
\end{align*}for all $q\in\mathbb{R}$ due to the fact that $V(\{\theta|\pi(\theta)L(\theta)=q\})=0$. Since $V(\{\theta|\pi(\theta)L(\theta)=\tilde{q}\})=0$ holds as well, it is also the case that\begin{align*}
    \lim_{\tilde{q}\downarrow q}P(\{\theta|\pi(\theta)L(\theta)>\tilde{q}\}|Y)&=1-\lim_{\tilde{q}\downarrow q}P(\{\theta|\pi(\theta)L(\theta)\leq\tilde{q}\}|Y)\\&=1-\lim_{\tilde{q}\downarrow q}P(\{\theta|\pi(\theta)L(\theta)<\tilde{q}\}|Y),
\end{align*} and we can use the same technique as before to show that \begin{align*}
    \lim_{\tilde{q}\downarrow q}P(\{\theta|\pi(\theta)L(\theta)>\tilde{q}\}|Y)=P(\{\theta|\pi(\theta)L(\theta)>q\}|Y).
\end{align*} Since we have just shown that the function $q\mapsto P(\{\theta|\pi(\theta)L(\theta)>q\}|Y)$ is continuous, and since it tends to $0$ as $q\to\infty$, to $1$ as $q\downarrow0$, it follows that the equation \begin{align}\label{eq:quantile_equation}
    P(\{\theta|\pi(\theta)L(\theta)>q|Y)=\alpha
\end{align} has at least one solution. In addition, the function is also strictly decreasing in the interval $(\inf_{\pi(\theta)L(\theta)>0}\pi(\theta)L(\theta),\sup_{\pi(\theta)L(\theta)>0}\pi(\theta)L(\theta))$,\\ since for any $\inf_{\pi(\theta)L(\theta)>0}\pi(\theta)L(\theta)<q_1<q_2<\sup_{\pi(\theta)L(\theta)>0}\pi(\theta)L(\theta)$\begin{align*}
    P(\{\theta|\pi(\theta)L(\theta)>q_1\}|Y)&-P(\{\theta|\pi(\theta)L(\theta)>q_2\}|Y)\\&=P(\{\theta|\pi(\theta)L(\theta)\in(q_1,q_2)\}|Y)>0,
\end{align*} because the set $\{\theta|\pi(\theta)L(\theta)\in(q_1,q_2)\}$ has non-zero volume due to the fact that $\pi(\theta)L(\theta)$ is continuous and non-constant. It follows that Equation \eqref{eq:quantile_equation} has only one solution, defined as $q_\alpha$. We can use the exact same technique to show that the volume $V(H_\alpha)$ is strictly increasing and continuous, since the only property of $P(\cdot|Y)$ that we used is that it is a measure, and the volume is also a measure.
\end{proof}

\subsection*{\textAtwo\label{ssec: Atwo}}

\begin{theorem}[Theorem 3.1 in the main article]
Suppose that $\theta^{(1)},\dots,\theta^{(T)}$ is ergodic and converges in distribution, as $T\to\infty$, to the symmetric posterior $p(\theta|Y)$ of a mixture model. Let $\theta^{(1)\star},\dots,\theta^{(T)\star}$ denote a relabelling of $\theta^{(1)},\dots,\theta^{(T)}$. For any fixed choice of $A$ for which $\pi\cdot L$ is bounded away from 0, in the sense that a $q>0$ exists such that $\pi(\theta)\cdot L(\theta)\geq q$ for all $\theta\in A$, any symmetric estimator of the form\begin{align*}
    \hat{Z}^{-1}_{\textup{S}}(G)=\frac{1}{G!}\sum_{o=1}^{G!}\frac{1}{T/2}\sum_{\substack{t=T/2+1,\\P_o(\theta^{(t)\star})\in A}}^T\frac{1/V(A)}{\pi(\theta^{(t)\star})L(\theta^{(t)\star})},
\end{align*} almost surely converges to the reciprocal marginal likelihood $Z^{-1}(G)$ as $T\to\infty$. If $\theta^{(1)},\dots,\theta^{(T)}$ is also ergodic of degree 2, then $\sqrt{T}(\hat{Z}^{-1}(G)-Z^{-1}(G))$ converges, as $T\to\infty$, in distribution to the normal distribution with mean equal to 0 and variance depending on the choice of $A$.
\end{theorem}

\begin{proof}
First, let us note that $\hat{Z}^{-1}_{\textup{S}}(G)$ is symmetric, which implies that it is not affected by changes in labels. Thus, this estimator computed on the relabelled sample is exactly equal to the same estimator on the same sample, that is not relabelled: \begin{align*}
    \hat{Z}^{-1}_{\textup{S}}(G)=\frac{1}{G!}\sum_{o=1}^{G!}\frac{1}{T/2}\sum_{\substack{t=T/2+1,\\P_o(\theta^{(t)})\in A}}^T\frac{1/V(A)}{\pi(\theta^{(t)})L(\theta^{(t)})},
\end{align*}

By exchanging the order of summation, \begin{align*}
    \hat{Z}^{-1}_{\textup{S}}(G)=\frac{1}{T/2}\sum_{t=T/2+1}^T\frac{\frac{1}{G!}\sum_{o=1}^{G!}\mathds{1}_{A}(P_o(\theta^{(t)}))/V(A)}{\pi(\theta^{(t)})L(\theta^{(t)})},
\end{align*}where $\mathds{1}$ denotes an indicator function.

Due to \citet[Theorems 3 and 4,][]{Ti94-MarkovChainsExploringPosteriors}, the statement of the theorem is correct if the expectation of $\frac{\frac{1}{G!}\sum_{o=1}^{G!}\mathds{1}_{A}(P_o(\theta))/V(A)}{\pi(\theta)L(\theta)}$ is equal to $Z^{-1}(G)$ if $\theta\sim p(\theta|Y)$. Symmetry of the posterior implies that $\theta\sim p(\theta|Y)$ is equivalent to $P_o(\theta)\sim p(\theta|Y)$ for any $o\in\{1,\dots,G!\}$. It follows that \begin{align*}
    E\left[\frac{\frac{1}{G!}\sum_{o=1}^{G!}\mathds{1}_{A}(P_o(\theta))/V(A)}{\pi(\theta)L(\theta)}\right]&=\frac{1}{G!}\sum_{o=1}^{G!}E\left[\frac{\mathds{1}_{A}(P_o(\theta))/V(A)}{\pi(\theta)L(\theta)}\right]\\&=\frac{1}{G!}\sum_{o=1}^{G!}E\left[\frac{\mathds{1}_{A}(\theta)/V(A)}{\pi(\theta)L(\theta)}\right]=E\left[\frac{\mathds{1}_{A}(\theta)/V(A)}{\pi(\theta)L(\theta)}\right].
\end{align*} Bayes' rule implies \begin{align*}
    E\left[\frac{\mathds{1}_{A}(\theta)/V(A)}{\pi(\theta)L(\theta)}\right]=\int_A \frac{p(\theta|Y)/V(A)}{\pi(\theta)L(\theta)}\;d\theta=\int_A \frac{p(\theta|Y)/V(A)}{Zp(\theta|Y)}\;d\theta=Z^{-1}(G).
\end{align*} This concludes the proof.
\end{proof}

A $1-p$ confidence interval can thus be calculated by combining the confidence interval of $\hat{V}(B_{\hat{\theta},\hat{\Sigma},c,\alpha})$ (which follows a scaled binomial distribution as it is proportional to a sum of i.i.d Bernoulli random variables) with the confidence interval of the above RIS estimator, which approximately follows a normal distribution, under mild assumptions on the MCMC sample. Since these estimators are independent, their confidence intervals can be combined using standard inequalities. Since we are usually interested in the logarithm of the THAMES, the logarithmic $\sqrt{1-p}$ confidence intervals of these independent estimators can simply be added to obtain a $1-p$ confidence interval. When we use the identity from \citet{No04-MixtureModelsDifferentSizeLink,No07-MargLikEmptyComponents} to recursively estimate the THAMES, the upper and lower bound are recursively estimated as well, using the same identity.

\subsection*{\textAthree\label{ssec: Athree}}

Optimality results were obtained in \citet{Me_et_al24-thames} in the case that $K=1$, under the assumption of normality of the posterior. We are going to present results in the case that the posterior is a mixture of normal distributions. This case is important because any density can be approximated by such a mixture, under few assumptions \citep{Ba10-approximate_by_Gaussian_mixtures}.

\begin{theorem}[Theorem 3.2 in the main article]
Consider a convex combination of truncated harmonic mean estimators, \begin{align}
    \hat{Z}^{-1}_{\text{Mix}}=\sum_{k=1}^K\omega_k\frac{1}{T/2}\sum_{\substack{t=T/2+1,\\\theta^{(t)}\in A_k}}^T\frac{1/V(A_k)}{\pi(\theta^{(t)})L(\theta^{(t)})},
\end{align} where $\omega_1,\dots,\omega_K$ are weights that sum to 1 and $A_1,\dots,A_K$ are truncation sets of the different estimators. Suppose that the posterior density $p(\theta|Y)$ is a Gaussian mixture, $p(\theta|Y)=\sum^K_{k=1}a_k\textup{MVN}_R(\theta;m_k,S_k),$ where $\text{MVN}_R(\theta;m_k,S_k)$ denotes the Gaussian law of dimension $R$ with mean $m_k$ and covariance matrix $S_k$. If the draws $\theta^{(1)},\dots,\theta^{(T)}$ from the posterior distribution are independent, an upper bound on the variance of $\hat{Z}^{-1}_{\text{Mix}}$, conditional on $Y$, is minimised by $A_k=E_{m_k,S_k,c_R}, \omega_k=a_k,$ where $c_R$ is the minimising constant determined in \citet{Me_et_al24-thames}, which is asymptotically equal to $\sqrt{R+1}$, i.e.,  $\lim_{R\to\infty}\frac{c_R}{\sqrt{R+1}}=1$.
\end{theorem}\begin{proof}
Because of independence, unbiasedness, the Cauchy-Schwartz and the QM-AM inequality, the variance of $\hat{Z}^{-1}_{\text{Mix}}$ is given by  \begin{align*}
    &\text{Var}[\hat{Z}^{-1}_{\text{Mix}}|Y]=\text{Var}\left[\left.\sum_{k=1}^K\omega_k\frac{1}{T/2}\sum_{\substack{t=T/2+1,\\\theta^{(t)}\in A_k}}^T\frac{1/V(A_k)}{\pi(\theta^{(t)})L(\theta^{(t)})}\right|Y\right]\stackrel{\text{i.i.d}}=\\
    &\frac{1}{T/2}\text{Var}\left[\left.\sum_{k=1}^K\omega_k\frac{\mathds{1}_{A_k}(\theta^{(1)})/V(A_k)}{\pi(\theta^{(1)})L(\theta^{(1)})}\right|Y\right]\stackrel{\text{unbiased}}=\\
    &\frac{1}{T/2}\text{E}\left[\left.\left(\sum_{k=1}^K\omega_k\frac{\mathds{1}_{A_k}(\theta^{(1)})/V(A_k)}{\pi(\theta^{(1)})L(\theta^{(1)})}\right)^2\right|Y\right]-\frac{Z^{-2}}{T/2}\stackrel{\text{linearity}}=\\
    &\frac{1}{T/2}\sum_{j,k=1}^K\omega_k\omega_j\text{E}\left[\left.\frac{\mathds{1}_{A_k\cap A_j}(\theta^{(1)})/(V(A_k)V(A_j))}{(\pi(\theta^{(1)})L(\theta^{(1)}))^2}\right|Y\right]-\frac{Z^{-2}}{T/2}\stackrel{\text{Cauchy Schwartz}}\leq
    \\&\frac{1}{T/2}\sum_{j,k=1}^K\sqrt{E\left[\left.\omega_k^2\frac{\mathds{1}_{A_k}(\theta^{(1)})/V(A_k)^2}{(\pi(\theta^{(1)})L(\theta^{(1)}))^2}\right|Y\right]E\left[\left.\omega_j^2\frac{\mathds{1}_{A_j}(\theta^{(1)})/V(A_j)^2}{(\pi(\theta^{(1)})L(\theta^{(1)}))^2}\right|Y\right]}\\&-\frac{Z^{-2}}{T/2}\stackrel{\text{QM-AM}}\leq
    \\&\frac{K^2}{T/2}\sqrt{\frac{1}{K^2}\sum_{j,k=1}^K\omega_k^2\omega_j^2E\left[\left.\frac{\mathds{1}_{A_k}(\theta^{(1)})/V(A_k)^2}{(\pi(\theta^{(1)})L(\theta^{(1)}))^2}\right|Y\right]E\left[\left.\frac{\mathds{1}_{A_j}(\theta^{(1)})/V(A_j)^2}{(\pi(\theta^{(1)})L(\theta^{(1)}))^2}\right|Y\right]}\\&-\frac{Z^{-2}}{T/2}=\\&\frac{K^2}{T/2}\left(\frac{1}{K}\sum_{k=1}^K\omega_k^2E\left[\left.\frac{\mathds{1}_{A_k}(\theta^{(1)})/V(A_k)^2}{(\pi(\theta^{(1)})L(\theta^{(1)}))^2}\right|Y\right]\right)-\frac{Z^{-2}}{T/2}=\\&\frac{K}{T/2}\sum_{k=1}^K\omega_k^2E\left[\left.\frac{\mathds{1}_{A_k}(\theta^{(1)})/V(A_k)^2}{(\pi(\theta^{(1)})L(\theta^{(1)}))^2}\right|Y\right]-\frac{Z^{-2}}{T/2}.
\end{align*}

We can use the fact that $p(\theta|Y)\geq a_k\text{MVN}_R(\theta;m_k,S_k)$ for any $k$ to get \begin{align*}
    &\frac{K}{T/2}\sum_{k=1}^K\omega_k^2E\left[\left.\frac{\mathds{1}_{A_k}(\theta^{(1)})/V(A_k)^2}{(\pi(\theta^{(1)})L(\theta^{(1)}))^2}\right|Y\right]-\frac{Z^{-2}}{T/2}=\\&\frac{K}{T/2}\sum_{k=1}^K\omega_k^2\int_{A_k}\frac{p(\theta|Y)/V(A_k)^2}{(\pi(\theta)L(\theta))^2}\;d\theta-\frac{Z^{-2}}{T/2}=\\&\frac{K}{T/2}\sum_{k=1}^K\omega_k^2\int_{A_k}\frac{p(\theta|Y)/V(A_k)^2}{(Zp(\theta|Y))^2}\;d\theta-\frac{Z^{-2}}{T/2}=\\&\frac{KZ^{-2}}{T/2}\sum_{k=1}^K\omega_k^2\int_{A_k}\frac{1/V(A_k)^2}{p(\theta|Y)}\;d\theta-\frac{Z^{-2}}{T/2}\leq\\&\frac{KZ^{-2}}{T/2}\sum_{k=1}^K\frac{\omega_k^2}{a_k}\int_{A_k}\frac{1/V(A_k)^2}{\text{MVN}_R(\theta;m_k,S_k)}\;d\theta-\frac{Z^{-2}}{T/2}.
\end{align*} $\text{MVN}_R(\theta;m_k,S_k)$ is itself a density, and the integral above is exactly the integral found in the proof of Theorem 1, Equation \eqref{eq: var_thames}. Thus, this integral is proportional to the variance of a THAMES estimator for which the posterior density is exactly normal. This case has been dealt with in \citet[Theorem 1]{Me_et_al24-thames}, and the integral is indeed minimised by $E_{m_k,S_k,c_R}$. The transformation formula implies that the integral is independent of $k$ in this case, since \begin{align*}
    &\frac{K}{T/2}\sum_{k=1}^K\frac{\omega_k^2}{a_k}\int_{E_{m_k,S_k,c_R}}\frac{Z^{-2}/V(E_{m_k,S_k,c_R})^2}{\text{MVN}_R(\theta;m_k,S_k)}\;d\theta-\frac{Z^{-2}}{T/2}=\\&\frac{K}{T/2}\sum_{k=1}^K\frac{\omega_k^2}{a_k}\int_{E_{0_R,I_R,c_R}}\frac{Z^{-2}/V(E_{0_R,I_R,c_R})^2}{\text{MVN}_R(\theta;0_R,I_R)}\;d\theta-\frac{Z^{-2}}{T/2}=\\&\frac{K}{T/2}\int_{E_{0_R,I_R,c_R}}\frac{Z^{-2}/V(E_{0_R,I_R,c_R})^2}{\text{MVN}_R(\theta;0_R,I_R)}\;d\theta\left(\sum_{k=1}^K\frac{\omega_k^2}{a_k}\right)-\frac{Z^{-2}}{T/2},
\end{align*} and the quadratic mean is minimised by $\omega_k=a_k$.
\end{proof}

\subsection*{\textAfour\label{ssec: Afour}}

\begin{theorem}[Theorem 3.1 in the main article]
Let $\psi_1,\dots,\psi_{G!}$ denote the $G!$ different ordering algorithms that map $\theta$ to different permutations of the ordering constraint \begin{align*}
    W(\xi_1)\leq\dots\leq W(\xi_G),
\end{align*} respectively, by permuting its labels. Let $\Delta$ denote the adjacency matrix defined by $\Delta_{g_1,g_2}=1$ if $W(\xi_{g_1})< W(\xi_{g_2})$ for all $\xi_{g_1},\xi_{g_2}$ that satisfy the ellipsoidal constraint on $E_{\hat{\theta},\hat{\Sigma},c}$, and let $\Omega$ denote any index set which includes the collection of all topological orderings on the graph defined by $\Delta$. The THAMES is equal to \begin{align}
     \hat{Z}^{-1}_{\textup{THAMES}}(G)=\frac{1}{G!}\sum_{o\in \Omega}\frac{1}{T/2}\sum_{\substack{t=T/2+1\\\psi_o(\theta^{(t)\star})\in B_{\hat{\theta},\hat{\Sigma},c,\alpha}}}^T\frac{1/V(B_{\hat{\theta},\hat{\Sigma},c,\alpha})}{\pi(\theta^{(t)\star})L(\theta^{(t)\star})}.
\end{align}
\end{theorem}

\begin{proof}
The key identity for this proof is that for every $t$ and for every function $\iota$ of $\theta$\begin{align*}
    \sum_{o=1}^{G!}\iota(P_o(\theta^{(t)\star}))=\sum_{o=1}^{G!}\iota(\psi_o(\theta^{(t)\star})).
\end{align*}It is true because the functions $\psi_o,P_o$ each assign a different permutation to $\theta$. This can be illustrated via Figure \ref{fig:dummy_example_supp}: let $T=10$, $\theta^{(t)\star}=(\mu^{(t)\star}_1,\mu^{(t)\star}_2)$, and $\iota(\theta^{(t)\star})=\mathds{1}_{E_{\hat{\theta},\hat{\Sigma},c}}(\theta^{(t)\star})$. In this case, we are simply counting the number of times that $P_o(\theta^{(t)\star})$ is in $E_{\hat{\theta},\hat{\Sigma},c}$. This can be done by applying $P_1,P_2$ one by one and counting each time (plots A and B), or by ordering every parameter such that $\mu_1^{(t)}>\mu_2^{(t)}$ via $\psi_1$, and then switching the order to $\mu_2^{(t)}>\mu_1^{(t)}$ via $\psi_2$ (plots C and D). In fact, $\psi_2$ can be avoided entirely in Figure \ref{fig:dummy_example_supp} since $E_{\hat{\theta},\hat{\Sigma},c}$ does not cross the identity line.

Since $\hat{Z}^{-1}_{\text{\text{THAMES}}}$ is also a sum over all permutations of $\theta^{(t)}$, we can replace $P_o$ by $\psi_o$ to get \begin{align*}
    \hat{Z}^{-1}_{\text{THAMES}}(G)=\frac{1}{G!}\sum_{o=1}^{G!}\frac{1}{T/2}\sum_{\substack{t=T/2+1,\\\psi_o(\theta^{(t)\star})\in B_{\hat{\theta},\hat{\Sigma},c,\alpha}}}^T\frac{1/V(B_{\hat{\theta},\hat{\Sigma},c,\alpha})}{\pi(\theta^{(t)\star})L(\theta^{(t)\star})}.
\end{align*}

Suppose that $\psi_o$ is such that $W(\xi_{g_1}^{\psi_o})> W(\xi_{g_2}^{\psi_o})$, but $\Delta_{g_1,g_2}=1$. This immediately implies that $\psi_o(\theta^{(t)\star})$ will not be included in the sum for any $\theta^{(t)\star}$, since  $W(\xi_{g_1})\leq  W(\xi_{g_2})$ for all $\xi_{g_1},\xi_{g_2}$ satisfying the ellipsoidal constraint on  $E_{\hat{\theta},\hat{\Sigma},c}$, and thus in particular in $B_{\hat{\theta},\hat{\Sigma},c,\alpha}$. Thus, the sum only includes orderings which respect the partial ordering implied by $\Delta$. This is how topological orderings are defined. Since $\Omega$ only excludes those parts of the series that would be excluded anyway, this finishes the proof.
\end{proof}

\begin{figure}
    \centering
    
    \includegraphics[width=\textwidth]{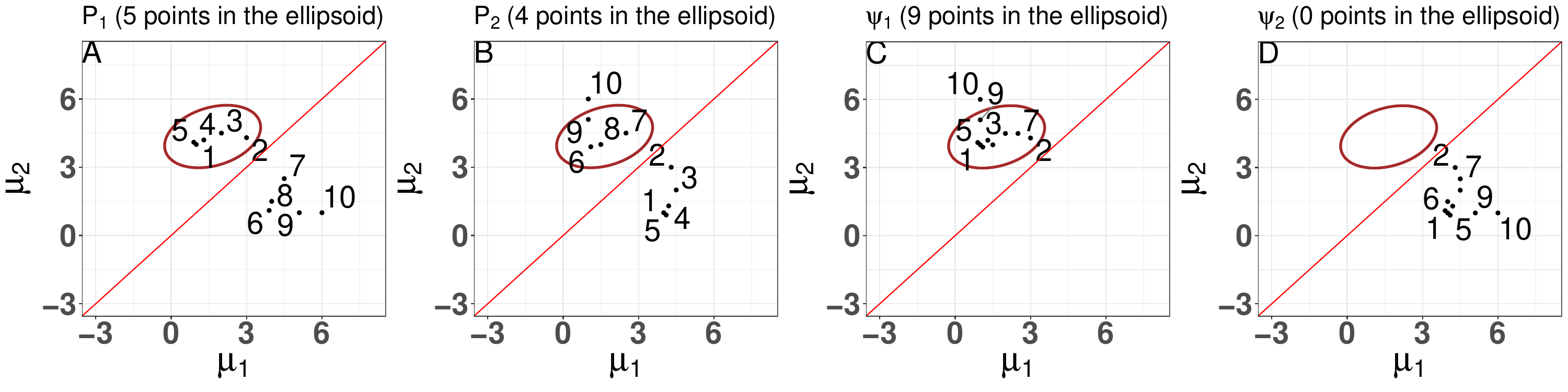}

    \caption{A toy example to illustrate the computation of the THAMES; the points that lie in the ellipsoid can be counted by applying $P_1,P_2$ or $\psi_1,\psi_2$.}
    \label{fig:dummy_example_supp}
\end{figure}

\subsection*{\textAfive\label{ssec: Afive}}

\begin{theorem}[Theorem 3.4 in the main article]\label{thm-nobile_implication_supp}
Suppose that a Dirichlet prior with parameters $e_1,\dots,e_G>0$ is used for the proportions and that $\hat{Z}(G-1)$ is a consistent estimator of $Z(G-1)$. If $\xi_1,\dots,\xi_G$ are a priori independent and their prior distribution does not depend on the total number of components, then\begin{align}\label{eq:multivariate_reduction_estim_supp}
    \hat{Z}(G)=\frac{\Gamma\left(\sum^G_{g=1}e_g\right)\Gamma\left(n+\sum^{G-1}_{g=1}e_g\right)}{\Gamma\left(n+\sum^G_{g=1}e_g\right)\Gamma\left(\sum^{G-1}_{g=1}e_g\right)}\cdot \hat{Z}(G-1)\cdot\frac{1}{\hat{p}_0(G)},
\end{align} with \begin{align*}
    \hat{p}_0(G)=\frac{1}{G}\sum^G_{g=1}\frac{1}{T}\sum_{t=1}^T\prod_{i=1}^n(1-\hat{z}_{i,g}^{(t)}),\quad\hat{z}_{i,g}^{(t)}=\frac{\tau_g^{(t)}f(Y_i;\xi_g^{(t)})}{\sum_{\tilde{g}=1}^G\tau_{\tilde{g}}^{(t)}f(Y_i;\xi_{\tilde{g}}^{(t)})}
\end{align*}is a consistent estimator of $Z(G)$ as long as $\hat{p}_0(G)$ is a consistent estimator of its posterior mean.
\end{theorem}

\begin{proof}

The data has the following latent variable representation: \begin{align*}
    C_i|\tau&\stackrel{\text{i.i.d}}\sim\text{Multinom}_G(1,\tau),\\Y_i|\xi,C_{i,g}=1&\sim f(\cdot;\xi_g),
\end{align*}with the hidden matrix of allocation vectors $C=(C_1,\dots,C_n)\in\mathcal{M}_{n\times G}(\{0,1\}).$

 \citet{No04-MixtureModelsDifferentSizeLink,No07-MargLikEmptyComponents} showed that \begin{align}\label{eq: phat0_expectation}
     \frac{Z(G)}{Z(G-1)}=\frac{\Gamma\left(\sum^G_{g=1}e_g\right)\Gamma\left(n+\sum^{G-1}_{g=1}e_g\right)}{\Gamma\left(n+\sum^G_{g=1}e_g\right)\Gamma\left(\sum^{G-1}_{g=1}e_g\right)}\cdot\frac{1}{1-p(\exists i:C_{i,G}=1|Y)}.
 \end{align}
 
 As pointed out by  \citet{No07-MargLikEmptyComponents},\begin{align*}
    1-p(\exists i:C_{i,G}=1)=p\left(\left.\sum_{i=1}^nC_{i,G}=0\right|Y\right),
\end{align*} the posterior probability that component $G$ is empty. By the law of total probability\begin{align*}
    p\left(\left.\sum_{i=1}^nC_{i,G}=0\right|Y\right)=\int p\left(\left.\sum_{i=1}^nC_{i,G}=0\right|Y,\theta\right)p(\theta|Y)\;d\theta
\end{align*} The distribution of $C_{i}$ given $Y,\theta$ is i.i.d multinomial with probabilities
\begin{align*}
    z_{i,g}=\frac{\tau_g f(Y_i,\xi_g)}{\sum_{\tilde{g}=1}^G\tau_{\tilde{g}} f(Y_i;\xi_{\tilde{g}})},\quad g=1,\dots,G.
\end{align*} This simplifies the integral to \begin{align*}
    \int p\left(\left.\sum_{i=1}^nC_{i,G}=0\right|Y,\theta\right)p(\theta|Y)\;d\theta=\int \prod_{i=1}^n(1-z_{i,G})p(\theta|Y)\;d\theta,
\end{align*} which is in turn equal to the posterior expectation of $\prod^n_{i=1}(1-z_{i,G})$. By the law of large numbers, \begin{align}\label{eq: phat0_nonsymmetric}
    \frac{1}{T}\sum_{t=1}^T\prod^n_{i=1}(1-\hat{z}_{i,G}^{(t)})
\end{align} is a valid estimator of this expectation. Due to the symmetry of the posterior, choosing any other component $g=1,\dots,G$ instead of the last component $G$ would give the exact same result. Averaging over all of these versions produces a symmetric estimator. This estimator is exactly equal to $\hat{p}_0$. Combining the unbiasedness of $\hat{p}_0$ with the consistency of $\hat{Z}(G-1)$ gives the assertion.
 
\end{proof}

\section*{\suppBtext \label{sec: suppB}}

\subsection*{\textBone \label{ssec: Bone}}

We choose $c=\sqrt{R+1}$, following the results of \citet{Me_et_al24-thames}. The hyperparameter $\alpha$ could be chosen via grid-search over $(0,1)$ such that it minimises the empirical variance of the THAMES. However, while this approach may work, it is computationally expensive, since we would have to approximate the volume of $B_{\hat{\theta},\hat{\Sigma},c,\alpha}$ via a new Monte Carlo sample for each value that $\alpha$ can take on the grid. Instead, we suggest choosing $\alpha$ such that the posterior distribution is well-behaved on $B_{\hat{\theta},\hat{\Sigma},c,\alpha}$. We take this to mean that the distribution of the sample from the posterior within $B_{\hat{\theta},\hat{\Sigma},c,\alpha}$ is as close as possible to a normal distribution, truncated such that at most 50 percent of the posterior sample is within $B_{\hat{\theta},\hat{\Sigma},c,\alpha}$. This is because the tails of the posterior can cause difficulties in reciprocal importance sampling estimators and the tails of the normal distribution are particularly well behaved in the case of the THAMES if they are cut off at 50 percent, as was shown in \citet{Me_et_al24-thames}.

If the posterior is normal, the negative unnormalised log-posterior values \begin{align*}\eta^{(1)},\dots,\eta^{(T)}=
    -\log(\pi(\theta^{(1)})L(\theta^{(1)})),\dots,-\log(\pi(\theta^{(T)})L(\theta^{(T)}))
\end{align*} follow, if they are truncated to an HPD region, a truncated, shifted and scaled $\chi^2$-distribution with $d$ degrees of freedom. This is because the quadratic form \begin{align*}
    (\theta^{(t)}-m)^\intercal S^{-1}(\theta^{(t)}-m)
\end{align*} follows a $\chi^2$-distribution if $m$ and $S$ are the true posterior mean and covariance matrix, respectively.

We set $\alpha$ such that it minimises the Kolmogorov distance against this $\chi^2$-distribution. Let $F$ denote the empirical cumulative distribution function (ECDF) of $\eta^{(1)},\dots,\eta^{(T)}$ conditioned on the $\alpha$-HPD region $H_{\alpha}$,
\begin{align*}
    F(\eta)=\frac{1}{\sum^T_{t=1}\mathds{1}_{H_{\alpha}}(\theta^{(t)})}\sum^T_{t=1}\mathds{1}_{H_{\alpha}}(\theta^{(t)})\cdot\mathds{1}_{(-\infty,\eta]}(\eta^{(t)}),
\end{align*} 
with $\mathds{1}$ denoting the indicator function, and let $\chi_{M,v}^2$ denote the cumulative distribution function (CDF) of the shifted and scaled $\chi^2$ distribution truncated on $[0,\max \{\eta^{(1)},\dots,\eta^{(T)}\}]$ with $d$ degrees of freedom, mean $M$ and variance $v$. 
The Kolmogorov distance \citep[see for instance][]{WiRa11-kolmogorov_distance} on $H_{\alpha}$ is given by \begin{align}\label{eq: kolm_test_statistic}
    \sup_{\theta\in H_{\alpha}} \left|\chi^2_{\hat{M}_\alpha,\hat{v}_\alpha}(-\log(\pi(\theta)L(\theta)))-F(-\log(\pi(\theta)L(\theta)))\right|,
\end{align} where the mean and variance parameter are computed using the moment-estimators $\hat{M}_\alpha,\hat{v}_\alpha$ on the values of $\eta^{(1)},\dots,\eta^{(T)}$ for which the corresponding parameters are within $H_{\alpha}$. We minimise this distance by evaluating Equation \eqref{eq: kolm_test_statistic} on some subset of the grid $\{\theta^{(1)},\dots,\theta^{(T)}\}$ for all sensible values of $\alpha$, with $\alpha$ being capped at 50 percent for the reasons that we mentioned before (for more details on this technique, see the following subsection).

\paragraph{\textbf{Summary}}Choosing $\alpha$ such that it minimises Equation \eqref{eq: kolm_test_statistic} has many advantages. An optimal value is easy to compute, since it can be efficiently approximated by evaluating at most $T^2$ different values due to the discreteness of $F$. It also conforms with previous results from \citet{Me_et_al24-thames}: if the posterior is exactly normal, the Kolmogorov distance is going to approach $0$ if $\alpha=1$, which implies $B_{\hat{\theta},\hat{\Sigma},c,\alpha}=E_{\hat{\theta},\hat{\Sigma},c}$. This means that the asymptotically optimal set derived by \citet{Me_et_al24-thames} for the case that the posterior is exactly normal is also going to be chosen by our method if this case occurs. Finally, if the posterior is not normal, our method chooses $\alpha$ such that the THAMES is well-behaved, in the sense that the Kolmogorov distance is going to penalise extreme outliers.

Even without this optimisation algorithm, we have observed that choosing $c=\sqrt{R+1}$ and $\alpha=0.5$ lead to good results, which is in congruence with similar observations from \citet{Me_et_al24-thames} and \citet{Re20-newthames_variation}.

\subsection*{\textBtwo\label{ssec: Btwo}}

We want to choose $\alpha$ such that it minimises the Kolmogorov distance \begin{align*}\sup_{\theta\in H_{\alpha}} \left|\chi^2_{\hat{M}_\alpha,\hat{s}_\alpha}(-\log(\pi(\theta)L(\theta)))-F(-\log(\pi(\theta)L(\theta)))\right|.
\end{align*} We do this by first choosing a grid for the different values of $\alpha$, and then, for each value on the grid, approximating the supremum via another grid search over $\theta$, respectively. 

Regarding the choice of the grid for $\alpha$, it is important to note that the above function is itself an estimator of the distance between the true CDF of the normalised log posterior values, approximated by $F$, and the CDF of the $\chi^2$-distribution. For this reason, we have experienced that small values of $\alpha$ caused the function to be very volatile and unreliable, since almost all of the normalised log posterior values are excluded in this case. We have experienced that this can be avoided by only evaluating values of $\alpha$ that are larger than $0.2$.

In addition, the function $F$ is constant for most values of $\theta$, excluding the values of the posterior sample itself, $\{\theta^{(1)},\dots,\theta^{(T)}\}$. In turn, this leaves us with a possibility of $T$ different values of $\alpha$ to check, but we chose to thin the sequence by a step-size of 100, since we made the experience that this greatly reduces computation time with no loss in the efficiency of the THAMES. Regarding the choice of the grid for $\theta$, we chose the full grid, since the computation time was quite small, even for a large grid. If the chosen $\alpha$ is larger than 0.5, it is then truncated to 0.5, to be in accordance with \citet{Me_et_al24-thames}.

\subsection*{\textBthree \label{ssec: Bthree}}

$\xi_{g_1}=\xi_{g_2}$ lies in $E_{\hat{\theta},\hat{\Sigma},c}$ if, and only if, the minimum\begin{align}\label{eq:quadopt}
    \text{min}_{\theta:\xi_{g_1}=\xi_{g_2}}(\theta-\hat{\theta})^\intercal\hat{\Sigma}^{-1}(\theta-\hat{\theta})
\end{align} is smaller or equal to $c^2$. This is a quadratic optimisation problem, which can be solved in polynomial time since $\hat{\Sigma}^{-1}$ is positive definite, for example by using one of the algorithms given in \citet{GoId83-quadratic_optimization02,GoId06-quadratic_optimization01}. We do this using the quadprog R package \citep{BeWe19-quadprog}.

\subsection*{\textBfour\label{ssec: Bfour}}

The matrix $\Delta$ is constructed as follows: if it has been established that $\xi_{g_1}=\xi_{g_2}$ for $g_1,g_2$ and for some pair of values $(\xi_{g_1},\xi_{g_2})$ that follows the ellipsoidal constraint given by $E_{\hat{\theta},\hat{\Sigma},c}$, then $\Delta_{g_1,g_2}=0$ since $W(\xi_{g_1})=W(\xi_{g_2})$ in that case. If not, then we test $W(\xi_{g_1,\text{Unif}}^{(t)})<W(\xi_{g_2,\text{Unif}}^{(t)})$ with $(\xi_{g_1,\text{Unif}}^{(t)},\xi_{g_2,\text{Unif}}^{(t)})$ being from the Monte Carlo sample that is simulated uniformly on $E_{\hat{\theta},\hat{\Sigma},c}$. If there exists a $j$ such that the condition is false, then $\Delta_{g_1,g_2}$ is set to 0. It is set to 1 otherwise. The law of large numbers guarantees $\Delta_{g_1,g_2}$ will be evaluated accurately if the size of the Monte Carlo sample $N$ is large enough.

The choice of $\Omega$ that is used in this article is constructed via an algorithm that returns all topological orderings on the graph defined by $\Delta$. An algorithm that does this in polynomial time with respect to the size of $\Omega$ is given in \citet{KnSz74-alltoporderings}. We implemented a version of this algorithm.

A very rough upper bound on the size of $\Omega$ is given by $G!/|I'(G)|!$, where $|I'(G)|$ denotes the number of vertices of the longest path in the graph generated by $\Delta$. This is because there are $|I'(G)|+1$ possibilities of choosing the position of the first node that is not in $I'(G)$, $|I'(G)|+2$ possibilities of choosing the position of the second node and so on, and the order of the nodes within $I'(G)$ is fixed. Note that $I'(G)$ can be computed by determining the shortest path of a graph with adjacency matrix $-\Delta$, by applying a shortest-path algorithm that can handle negative edges, such as the Bellman–Ford algorithm \citep{Fo56-NetworkFlowTheory,Be58-OnARoutingProblem}. This is possible since $\Delta$ cannot have cycles, since the relation $W(\xi_{g_1})<W(\xi_{g_2})$ is transitive. We use the shortest.path function from the igraph package \citep{igraph03,igraph02,igraph01} for this task.

Whenever $G!/|I'(G)|!$ passes 50,000, we take this as a cue that the THAMES cannot be evaluated in a reasonable amount of time. In this case, we iteratively shrink the ellipsoid until the bound passes 50,000. This is done by repeatedly dividing $c$ by 2 and, if the ellipsoid ends up containing no samples due to this shrinkage, setting $\hat{\theta}$ to the mode of the second half of the posterior sample and restarting the process. This makes sure that the ellipsoid always contains at least one point, the mode of the second half of the posterior sample. To make sure that $W$ can differentiate this point after $\hat{\theta}$ has been shifted, we train QDA only on those sample points that are within the ellipsoid in this case. Note that $\Omega$ can also be chosen such that it allows an exact computation of the THAMES. We do not use this choice, since we consider the approximate choice presented above to be sufficiently precise, and much simpler than the exact one. The exact choice is still explained here to showcase that an exact computation of the THAMES is possible in a feasible amount of time.

\paragraph{An exact computation of $\Omega$}

A different choice of $\Omega$ can be computed by determining $\Delta$ via a variety of inequalities. If these inequalities hold for a given value of $g_1,g_2$, they imply $W(\xi_{g_1})< W(\xi_{g_2})$ for all values of $(\xi_{g_1},\xi_{g_2})$ that follow the ellipsoidal constraint given by $E_{\hat{\theta},\hat{\Sigma},c}$. Thus, these inequalities can be used to construct $\Delta$ like in the previous case, but without using the approximate verification through the Monte Carlo sample. This construction of $\Delta$ will thus give a choice of $\Omega$ through which the THAMES can be computed exactly, and not just approximately. In the following text, it is explained how these inequalities can be derived.

We have $W(\xi_{g_1})< W(\xi_{g_2})$ if, and only if \begin{align*}
    \hat{g}(\xi_{g_1}|I(G))+1-\hat{w}_{\hat{g}(\xi_{g_1}|I(G))}(\xi_{g_1}|I(G))< \hat{g}(\xi_{g_2}|I(G))+1-\hat{w}_{\hat{g}(\xi_{g_2}|I(G))}(\xi_{g_2}|I(G))
\end{align*} Suppose that $\hat{g}(\xi_{g_1}|I(G))$ and $\hat{g}(\xi_{g_2}|I(G))$ are different natural numbers. In this case, the above inequality holds if, and only if \begin{align}
    \hat{g}(\xi_{g_1}|I(G))< \hat{g}(\xi_{g_2}|I(G)),
\end{align} since the remainders of the sums are within $(0,1)$. $\hat{g}(\xi_{g_1}|I(G))$ and $\hat{g}(\xi_{g_2}|I(G))$ are different natural numbers if, and only if \begin{align*}
    \text{argmax}_{g\in I(G)}\hat{w}_g(\xi_{g_1}|I(G))\neq \text{argmax}_{g\in I(G)}\hat{w}_g(\xi_{g_2}|I(G)).
\end{align*} this is the case if, and only if\begin{align*}
    \text{argmax}_{g\in I(G)}\frac{\hat{w}_{g}(\xi_{g_1})}{\sum_{\tilde{g}\in I(G)}\hat{w}_{\tilde{g}}(\xi_{g'})}\neq \text{argmax}_{g\in I(G)}\frac{\hat{w}_{g}(\xi_{g_2})}{\sum_{\tilde{g}\in I(G)}\hat{w}_{\tilde{g}}(\xi_{g'})}.
\end{align*} Since the denominators are the same, this is equivalent to \begin{align*}
    \text{argmax}_{g\in I(G)}\hat{w}_{g}(\xi_{g_1})\neq \text{argmax}_{g\in I(G)}\hat{w}_{g}(\xi_{g_2}).
\end{align*} Since $\sum_g \hat{w}_{g}=1,\hat{w}_{g}>0\forall g$, this inequality holds if for all values of $\xi_{g_1},\xi_{g_2}$ within the ellipsoid there exists a number $g_1'$ such that $\hat{w}_{g_1'}(\xi_{g_1})>0.5$ and another number $g_2'>g_1'$ such that $\hat{w}_{g_2'}(\xi_{g_2})>0.5$. The inequality holds because the former statement implies that $\text{argmax}_{g\in I(G)}\hat{w}_{g}(\xi_{g_1})=g_1'$ and $\text{argmax}_{g\in I(G)}\hat{w}_{g}(\xi_{g_2})=g_2'$ for all values of $\xi_{g_1},\xi_{g_2}$ within $B_{\hat{\theta},\hat{\Sigma},c,\alpha}$.

By definition, $\hat{w}_{g_1'}(\xi_{g_1})>0.5$ if, and only if\begin{align*}
    &\frac{\frac{1}{G}\text{MVN}_u(\xi_{g_1};\hat{\xi}_{g_1'},\hat{S}_{g_1})}{\sum_{\tilde{g}=1}^G\frac{1}{G}\text{MVN}_u(\xi_{g_1};\hat{\xi}_{\tilde{g}},\hat{S}_{\tilde{g}})}>0.5\Leftrightarrow\\
    &\text{MVN}_u(\xi_{g_1};\hat{\xi}_{g_1'},\hat{S}_{g_1'})>0.5\sum_{\tilde{g}=1}^G\text{MVN}_u(\xi_{g_1};\hat{\xi}_{\tilde{g}},\hat{S}_{\tilde{g}})\Leftrightarrow\\&\text{MVN}_u(\xi_{g_1};\hat{\xi}_{g_1'},\hat{S}_{g_1'})>\sum_{\tilde{g}\neq g_1'}^G\text{MVN}_u(\xi_{g_1};\hat{\xi}_{\tilde{g}},\hat{S}_{\tilde{g}}).
\end{align*} For this to be true, it suffices that \begin{align*}
    \text{MVN}_u(\xi_{g_1};\hat{\xi}_{g_1'},\hat{S}_{g_1'})>(G-1)\text{MVN}_u(\xi_{g_1};\hat{\xi}_{\tilde{g}},\hat{S}_{\tilde{g}})
\end{align*} for all $\tilde{g}\neq g_1'$, since in that case the sum is dominated by this maximum. Reformulating this inequality using logarithms gives\begin{align*}
    &\log\text{MVN}_u(\xi_{g_1};\hat{\xi}_{g_1'},\hat{S}_{g_1'})>\log(G-1)+\log\text{MVN}_u(\xi_{g_1};\hat{\xi}_{\tilde{g}},\hat{S}_{\tilde{g}})\Leftrightarrow\\
    &\log|\hat{S}_{g_1'}|+(\xi_{g_1'}-\hat{\xi}_{g_1'})^\intercal\hat{S}_{g_1}^{-1}(\xi_{g_1'}-\hat{\xi}_{g_1'})+\log(G-1)<\log|\hat{S}_{\tilde{g}}|+(\xi_{\tilde{g}}-\hat{\xi}_{\tilde{g}})^\intercal\hat{S}_{\tilde{g}}^{-1}(\xi_{\tilde{g}}-\hat{\xi}_{\tilde{g}}).
\end{align*}

So, it suffices to check that \begin{align}\label{eq:max}
    &\max_{\theta=(\xi_1,\dots,\xi_G,\tau_1,\dots\tau_{G-1})\in E_{\hat{\theta},\hat{\Sigma},c}}\log|\hat{S}_{g_1'}|+(\xi_{g_1'}-\hat{\xi}_{g_1'})^\intercal\hat{S}_{g_1'}^{-1}(\xi_{g_1'}-\hat{\xi}_{g_1'})+\log(G-1)
\end{align} is strictly smaller than \begin{align}\label{eq:min}
    \min_{\tilde{g}\neq g_1',\theta=(\xi_1,\dots,\xi_G,\tau_1,\dots\tau_{G-1})\in E_{\hat{\theta},\hat{\Sigma},c}}\log|\hat{S}_{\tilde{g}}|+(\xi_{\tilde{g}}-\hat{\xi}_{\tilde{g}})^\intercal\hat{S}_{\tilde{g}}^{-1}(\xi_{\tilde{g}}-\hat{\xi}_{\tilde{g}}).
\end{align} for all values of $\tilde{g}\neq g_1'$. Thus, one can check if a value $g_1'$ exists for which this Expression \eqref{eq:max} is strictly smaller than \eqref{eq:min} for all values of $\tilde{g}$. Expressions \eqref{eq:max} and \eqref{eq:min} are quadratic optimality problems and since both $\hat{S}_{g_1'}$ and $\hat{S}_{\tilde{g}}$ are both positive definite, these problems can be solved in polynomial time, e.g., by determining the bounding box of $E_{\hat{\theta},\hat{\Sigma},c}$, translating this bounding box to simplex constraints, and by applying one of the algorithms of \citet{GoId83-quadratic_optimization02,GoId06-quadratic_optimization01}. If this inequality holds, and in addition a $g_1'<g_2'$ exists such that \begin{align}\label{eq:maxg2}
    &\max_{\theta=(\xi_1,\dots,\xi_G,\tau_1,\dots\tau_{G-1})\in E_{\hat{\theta},\hat{\Sigma},c}}\log|\hat{S}_{g_2'}|+(\xi_{g_2'}-\hat{\xi}_{g_2'})^\intercal\hat{S}_{g_2'}^{-1}(\xi_{g_2'}-\hat{\xi}_{g_2'})+\log(G-1)
\end{align} is strictly smaller than \begin{align}\label{eq:ming2}
    \min_{\tilde{g}\neq g_2',\theta=(\xi_1,\dots,\xi_G,\tau_1,\dots\tau_{G-1})\in E_{\hat{\theta},\hat{\Sigma},c}}\log|\hat{S}_{\tilde{g}}|+(\xi_{\tilde{g}}-\hat{\xi}_{\tilde{g}})^\intercal\hat{S}_{\tilde{g}}^{-1}(\xi_{\tilde{g}}-\hat{\xi}_{\tilde{g}}),
\end{align} then $W(\xi_{g_1})< W(\xi_{g_2})$ for all values $(\xi_{g_1},\xi_{g_2})$ that follow the ellipsoidal constraint given by $E_{\hat{\theta},\hat{\Sigma},c}$. Thus, an exact set $\Omega$ can be constructed as follows: \begin{itemize}
    \item[(1)] for every pair $(g_1,g_2)$, verify if there exists a pair $(g_1',g_2'),g_1'<g_2'$ such that the maximum in Equation \eqref{eq:max} is strictly smaller than the minimum in Equation \eqref{eq:min} and the maximum in Equation \eqref{eq:maxg2} is strictly smaller than the minimum in Equation \eqref{eq:ming2}. Set $\Delta_{g_1,g_2}=1$ if this is the case, $\Delta_{g_1,g_2}=0$ otherwise.
    
    \item[(2)] determine all elements of $\Omega$ by determining all topological orderings on $\Delta$, via the alogrithm of \citet{KnSz74-alltoporderings}.
\end{itemize}

\subsection*{\textBfive\label{ssec: Bfive}}

A full description of an approximation of the THAMES for mixture models is given in Algorithms \ref{algorithm: thamesmixalgo} and \ref{algorithm: thamesmixrecursion}. These algorithms have been constructed by combining the computational details given in Supplement B with the broader overview given in the main part of this chapter. Algorithm \ref{algorithm: thamesmixrecursion} recursively calls Algorithm \ref{algorithm: thamesmixalgo}, and is always used whenever the prior on the mixture proportions is Dirichlet, $\xi_1,\dots,\xi_G$ are a priori independent and the prior distribution does not depend on the total number of components. If these conditions are not fulfilled, Algorithm \ref{algorithm: thamesmixalgo} is used directly.

In addition, there are three edge cases in which the THAMES is undefined or cannot be computed in a feasible amount of time:\begin{itemize}
    \item[(1)] none of the second half of the posterior sample is included in the truncation set, in which case the THAMES is undefined. This can be avoided by setting the centre of the ellipsoid to the posterior mode of the second half of the posterior sample.
    
    \item[(2)] the volume estimator of the truncation set is equal to 0, in which case the THAMES is undefined. This can happen if the truncation ellipsoid does not contain enough of the posterior HPD region, in which case the probability of the Monte Carlo sample taking values in this region is too small. This can be avoided by decreasing the size of the ellipsoid, by, e.g., by dividing its radius by 2.
    
    \item[(3)] the permutation subset $\Omega$ is too large to allow a feasible computation time of the THAMES. This can also be avoided by decreasing the size of the ellipsoid, since this decreases the overlap of the mixture parameters within the ellipsoid and hence the size of $\Omega$.
\end{itemize}

A sufficient repetition of the modifications outlined above ensures that all three edge cases can be avoided, and that the THAMES can be computed.

{ \scriptsize \begin{algorithm}[H]
	\KwIn{MCMC sample $(\tau^{(t)},\xi^{(t)})_{t\in\llbracket 1,T\rrbracket}$, unnormalised posterior $\pi\cdot L$}
	
	\textbf{Full parameter:} $\theta^{(t)}_{t\in\llbracket 1,T\rrbracket}=(\xi_1^{(t)},\dots,\xi_G^{(t)},\tau_1^{(t)},\dots,\tau_{G-1}^{(t)})_{t\in\llbracket 1,T\rrbracket}$ with dimension $R$.
	
	\textbf{Sample splitting:} Calculate the empirical mean $\hat\theta$ and sample covariance matrix $\hat\Sigma$ based on the first $T/2$ posterior samples $(\theta^{(t)})_{t\in\llbracket 1,T/2\rrbracket}$.
	
	\textbf{Truncation ellipsoid:} $E_{\hat{\theta},\hat{\Sigma},c}=\{\theta|(\theta-\hat{\theta})^\intercal\hat{\Sigma}^{-1}(\theta-\hat{\theta})<c^2\}$ with $c=\sqrt{R+1}$.

\textbf{Uniform sample:} sample $\theta^{(1)}_{\text{Unif}},\dots,\theta^{(T)}_{\text{Unif}}\sim\text{Unif}(E_{\hat{\theta},\hat{\Sigma},c})$ such that $(\xi_{1,\text{Unif}}^{(t)},\dots,\xi_{G,\text{Unif}}^{(t)},\tau_{1,\text{Unif}}^{(t)},\dots,\tau_{G-1,\text{Unif}}^{(t)})_{t\in\llbracket 1,T\rrbracket}=\theta_{\text{Unif},t\in\llbracket 1,T\rrbracket}^{(t)}$.

\textbf{Truncation subset:}$B_{\hat{\theta},\hat{\Sigma},c,\alpha}=\{\theta|(\theta-\hat{\theta})^\intercal\hat{\Sigma}^{-1}(\theta-\hat{\theta})<c^2,\pi(\theta)L(\theta)>\hat{q}_\alpha\}$ where $\hat{q}_\alpha$ is the $1-\alpha$ quantile of $(\pi(\theta^{(t)})L(\theta^{(t)}))_{t\in\llbracket 1,T\rrbracket}$ and $\alpha$ is the value on a $(0.2,0.5)$ grid that minimises the Kolmogorov-Smirnov test statistic between the $\chi^2$ distribution and the empirical distribution of $-\log(\pi(\theta^{(1)})L(\theta^{(1)})),\dots, -\log(\pi(\theta^{(T)})L(\theta^{(T)}))$.

\textbf{Volume estimation:} $\hat{V}(B_{\hat{\theta},\hat{\Sigma},c,\alpha})=  V(E_{\hat{\theta},\hat{\Sigma},c})\cdot\frac{1}{N}\sum^N_{j=1}\mathds{1}_{\{\theta|\pi(\theta)L(\theta)>\hat{q}_\alpha\}}(\theta^{(j)}_{\text{Unif}}).$

\If{$\hat{V}(B_{\hat{\theta},\hat{\Sigma},c,\alpha})=0$}{Set $c\mapsto c/2$ and rerun the algorithm.}

\textbf{Overlap graph:} add an edge between $g_1,g_2\in\{1,\dots,G\}$ if $\{\theta|\xi_{g_1}=\xi_{g_2}\}\cap E_{\hat{\theta},\hat{\Sigma},c}\neq\emptyset$.

\textbf{Maximum independent set:} determine $I(G)$ on the overlap graph.

\textbf{Summary measure:} for all $g'\in\{1,\dots, G\}$ compute\begin{align*}
&\forall g\in\{1,\dots,G\},\; \hat{w}_{g}(\xi_{g'}^{(t)})=\frac{\frac{1}{G}\text{MVN}_u(\xi_{g'}^{(t)};\hat{\xi}_{g},\hat{S}_{g})}{\sum_{\tilde{g}=1}^G\frac{1}{G}\text{MVN}_u(\xi_{g'}^{(t)};\hat{\xi}_{\tilde{g}},\hat{S}_{\tilde{g}})},\hat{w}_{g}(\xi_{g'}^{(t)}|I(G))=\frac{\hat{w}_{g}(\xi_{g'}^{(t)})}{\sum_{\tilde{g}\in I(G)}\hat{w}_{\tilde{g}}(\xi_{g'}^{(t)})}\\&
\hat{g}(\xi_{g'}^{(t)}|I(G))=\text{argmax}_{g\in I(G)}\hat{w}_{g}(\xi_{g'}^{(t)}|I(G))\\&
W(\xi_{g'}^{(t)})=\hat{g}(\xi_{g'}^{(t)}|I(G))+1-\hat{w}_{\hat{g}(\xi_{g'}^{(t)}|I(G))}(\xi_{g'}^{(t)}|I(G)).
\end{align*}

\textbf{Adjacency matrix of partial orderings:} for all $g_1,g_2\in\{1,\dots,G\},\;\Delta_{g_1,g_2}=1$ if $W(\xi_{g_1,\text{Unif}}^{(t)})< W(\xi_{g_2,\text{Unif}}^{(t)})$ for all $t\in\{1,\dots,T\}$ and $\{\theta|\xi_{g_1}=\xi_{g_2}\}\cap E_{\hat{\theta},\hat{\Sigma},c}=\emptyset$.

\If{$G!/|I'(G)|!>50,000$, where $I'(G)$ is the longest path on $\Delta$}{Set $c\mapsto c/2$ and rerun the algorithm.}

\textbf{Ordering functions}: $\psi_1,\dots,\psi_{G!}$ such that all possible outputs fulfil\begin{align*}
    &W(\xi_1^{\psi_1})\leq W(\xi_2^{\psi_1})\leq\dots\leq W(\xi_G^{\psi_1})\\
    &W(\xi_2^{\psi_2})\leq W(\xi_1^{\psi_2})\leq\dots\leq W(\xi_G^{\psi_2})\\
    &\vdots\\
    &W(\xi_G^{\psi_{G!}})\leq W(\xi_{G-1}^{\psi_{G!}})\leq\dots\leq W(\xi_1^{\psi_{G!}}).
\end{align*}

\textbf{Permutation subset}: $o\in\Omega$ if $\psi_g$ is in the set of all topological orderings on $\Delta$.

\textbf{Calculate} \begin{equation*}
 \hat{Z}^{-1}_{\text{THAMES}}(G)=\frac{1}{G!}\sum_{o\in \Omega}\frac{1}{T/2}\sum_{\substack{t=T/2+1\\\psi_o(\theta^{(t)\star})\in B_{\hat{\theta},\hat{\Sigma},c,\alpha}}}^T\frac{1/\hat{V}(B_{\hat{\theta},\hat{\Sigma},c,\alpha})}{\pi(\theta^{(t)\star})L(\theta^{(t)\star})}.
\end{equation*}

\If{$\hat{Z}^{-1}_{\text{THAMES}}(G)=0$}{Set $\hat{\theta}$ to the mode of $(\theta^{(i)})_{i\in\llbracket T/2+1,T\rrbracket},c=\sqrt{R+1}$ and rerun the algorithm.}

\textbf{Output:} the estimate $\hat{Z}^{-1}_{\text{THAMES}}(G)$ of the reciprocal marginal likelihood
	\caption{THAMES for mixture models}
	\label{algorithm: thamesmixalgo}
	\end{algorithm}}
	
{ \small \begin{algorithm}[H]
	\KwIn{MCMC sample $(\tau^{(t)},\xi^{(t)})_{t\in\llbracket 1,T\rrbracket}$, hyperparameters $n_1^0,\dots,n_{G-1}^0$, $n\times d$ dataset $Y$, density $f$, unnormalised posterior $\pi\cdot L$}
	
	\textbf{Calculate} \begin{equation*}
 \hat{p}_0(G)=\frac{1}{G}\sum^G_{g=1}\frac{1}{T}\sum_{t=1}^T\prod_{i=1}^n(1-\hat{z}_{i,g}^{(t)})\text{, where }\hat{z}_{i,g}^{(t)}=\frac{\tau_g^{(t)}f(Y_i;\xi_g^{(t)})}{\sum_{\tilde{g}=1}^G\tau_{\tilde{g}}^{(t)}f(Y_i;\xi_{\tilde{g}}^{(t)})}.
\end{equation*}

\If{$\hat{p}_0(G)>1/T$ and $G>2$}{
\textbf{Generate} a new MCMC sample with respect to $G-1$ .

\textbf{Calculate} $\hat{Z}_{\text{THAMES}}(G-1)$ through Algorithm \ref{algorithm: thamesmixrecursion}.

\textbf{Calculate}\begin{equation*}
    \hat{Z}^{-1}_{\text{THAMES}}(G)=\left(\frac{\Gamma\left(\sum^G_{g=1}n_g^0\right)\Gamma\left(n+\sum^{G-1}_{g=1}n_g^0\right)}{\Gamma\left(n+\sum^G_{g=1}n_g^0\right)\Gamma\left(\sum^{G-1}_{g=1}n_g^0\right)}\cdot \hat{Z}_{\text{THAMES}}(G-1)\cdot\frac{1}{\hat{p}_0(G)}\right)^{-1}.
\end{equation*}
}
\Else{\textbf{Calculate} $\hat{Z}^{-1}_{\text{THAMES}}(G)$ through Algorithm \ref{algorithm: thamesmixalgo}.}

\textbf{Output:} the estimate $\hat{Z}^{-1}_{\text{THAMES}}(G)$ of the reciprocal marginal likelihood \caption{Recursive THAMES for mixture models}
	\label{algorithm: thamesmixrecursion}
	\end{algorithm}}
	
It should be noted that these modifications could potentially increase the variance of the THAMES, as pointed out in the discussion. However, in practice we observed that these modifications are typically only necessary when the parameters of the mixture components are overfitted, and that in this case the THAMES estimates that the marginal likelihood is very low. So, it appears that the performance of the THAMES as a model choice measure is not diminished by these edge cases, since it will still estimate a value of $G$ that does not overfit.

\section*{\suppCtext \label{sec: suppC}}

Note that the conditions of Theorem \ref{thm-nobile_implication_supp} are not met for the univariate model in Supplement {\color{blue}\hyperref[ssec: Cthree]{C3}}, so the theorem was not used for univariate data. In addition, we have noticed that a large amount of the ellipsoid needed to be truncated in the case that its centre was close to the edge of the parameter space, which happened with variance and covariance matrix parameters. To avoid this, these parameters where mapped to the real line and the prior was adjusted via the absolute value of the determinant of the Jacobian before evaluating the THAMES. The transformations used are described in \citet{Ca_et_al17-stan}.

\subsection*{\textCone \label{ssec: Cone}}

An analytic solution of the marginal likelihood was computed in the case that the number of parameters and data points is very small  \citep[the same idea was used in][]{Fr06-FiniteMixtureAndMarkovSwitchingModels}, such that the performance of the THAMES could be compared to other state-of-the-art estimators.

\paragraph{\textbf{Model}}

We consider a univariate Gaussian mixture model,\begin{align}\label{eq: univ_gauss_mix_supp}
    Y\in\mathbb{R}^{n},\quad  Y_{1},\dots,Y_{n}|\mu,\sigma,\tau\stackrel{\text{i.i.d}}\sim \sum_{g=1}^G\tau_g\mathcal{N}(\mu_g,\sigma^2_g),
\end{align} where only the parameters $\mu=(\mu_1,\dots,\mu_G)$ are unknown and $\tau=(\tau_1,\dots,\tau_G),\sigma=(\sigma_1,\dots,\sigma_G)$ are known. They are known in the sense that we put a Dirac density on these parameters, where a Dirac density is a discrete density which is only equal to 1 at exactly 1 point. Let $C=(C_1,\dots,C_n)\in\mathcal{M}_{n\times G}(\{0,1\})$ denote the matrix of hidden allocation vectors. An analytic solution was computed and a Gibbs sampler was constructed using the following theorem. It is proven in Supplement \hyperref[sec: suppD]{D}.

\begin{theorem}\label{thm-compute_marglik}
If the prior is set to the distribution with the following symmetric densities, \begin{align}\label{eq: simple_prior}
    \pi(\mu)=\prod_{g=1}^G\mathcal{N}(\mu_g;0,1),\quad\pi(\sigma_g)=\begin{cases}1&\sigma_g=1,\\0&\textup{else},\end{cases}\quad\pi(\tau_g)=\begin{cases}1&\tau_g=1/G,\\0&\textup{else,}\end{cases}
\end{align} then the marginal likelihood is given by \begin{align}\label{eq: exact_marglik}
    p(Y)&=\sum_{C_1,\dots,C_n}p(Y|C_1,\dots,C_n)p(C_1,\dots,C_n)\\&=\frac{1}{G^n}\sum_{C_1,\dots,C_n}\textup{MVN}_n\left(Y;0_n,I_n+O\right),\nonumber
\end{align} where $O$ is a matrix such that $O_{i,j}=1$ if $C_i=C_j$, $O_{i,j}=0$ else. The conditional distributions of $\mu$ and the latent components $C_1,\dots,C_n$ are \begin{align*}
    \mu_g|C,Y&\stackrel{\textup{ind.}}\sim \mathcal{N}\left(\frac{\sum_{i=1}^n C_{i,g}Y_i}{1+\sum_{i=1}^n C_{i,g}},\frac{1}{1+\sum_{i=1}^n C_{i,g}}\right)\\C_i|\mu,Y&\stackrel{\textup{ind.}}\sim    \textup{Multinom}_G\left(1,\begin{pmatrix}\frac{\exp\left(-\frac{1}{2}\left(Y_i-\mu_1\right)^2\right)}{\sum_{g=1}^G\exp\left(-\frac{1}{2}\left(Y_i-\mu_g\right)^2\right)},\\\vdots,\\\frac{\exp\left(-\frac{1}{2}\left(Y_i-\mu_G\right)^2\right)}{\sum_{g=1}^G\exp\left(-\frac{1}{2}\left(Y_i-\mu_g\right)^2\right)}\end{pmatrix}\right).
\end{align*}
 \end{theorem}

\paragraph{\textbf{Settings}}

In all of the following scenarios, we set $n=10$ such that the formula given in Theorem \ref{thm-compute_marglik} could be computed in a reasonable amount of time. In addition, even though our prior assumptions were symmetric, we wanted to test how our estimator performs in the case that the data is simulated from a distribution that is not symmetric, in the sense that its components have distinct means and proportions. For all scenarios, we simulated 50 different datasets from the fixed parameters $(\mu^\star,\sigma^\star,\tau^\star)$. We set $\sigma^\star=(1,\dots,1)$, $\tau^\star=(1/3,2/3)$ if $G=2$, $\tau^\star=(2/6,1/6,3/6)$ if $G=3$. The mean was varied, \begin{align*}
    \mu^\star=\begin{cases}(3,3)^\intercal+(2(1-\rho)+1)(-1,1)^\intercal&G=2,\\(3,3,3)^\intercal+(2(1-\rho)+1)(-1,0,1)^\intercal&G=3,\end{cases}\quad\rho=0,1/2,1.
\end{align*}This ensures linear interpolation between means of $(0,3,6)^\intercal,(0,6)^\intercal$, whose values are particularly well separated, and means of $(2,3,4)^\intercal,(2,4)^\intercal$, which are less separate. Note that the true proportion parameters are not included in the class of models that was fitted. This is not a problem since we are not assessing the performance of estimators of the parameters, but of estimators of the marginal likelihood.

\paragraph{\textbf{Sampling procedure}}

A sample from the posterior of size $12,000$ was produced via Gibbs sampling initialised in the midpoint of the data. The first $2,000$ samples were discarded (burn-in) and the remaining $T=10,000$ samples were used to compute the bridge sampling estimator, the estimator from \citet{Re20-newthames_variation} (using the same tuning parameters as the THAMES) as well as the THAMES on the relabelled posterior sample, where all estimators were defined with the prior from Equation \eqref{eq: simple_prior}. 
 
We tested three scenarios: in the first, $G=2$ and the data is simulated with $G=2$. In the second, we tested our performance in the case of overfitting by computing the marginal likelihood for $G=3$ on the data that was simulated  with $G=2$. Analogously, the third scenario corresponds to underfitting in the case that $G=3$. 

\subsection*{\textCtwo \label{ssec: Ctwo}}

In this section, the performance of the THAMES was assessed in the case that the number of components $G$ is large. The analytic expression given by Equation \eqref{eq: exact_marglik} is untractable in this case. However, it can be computed if the mixture components are sufficiently distinct from each other \citep[the same idea was used in][]{Fr06-FiniteMixtureAndMarkovSwitchingModels}. This is because in this case, the likelihood conditional on the allocations is equal to $p(Y|C)\simeq 0$ for all but one allocation $C^0$ (and its label-switched versions, which result in the same likelihood). Thus, the marginal likelihood reduces to \begin{align}\label{eq: exact_seperate_marglik}
    p(Y)\simeq G!\cdot p(Y|C^0)\cdot p(C^0),
\end{align} where the factor $G!$ comes from the different label-switched versions of $C^0$.

$C^0$ was estimated via the maximum aposteriori (MAP) estimator ($\hat{C}^{0}_{i,g}=1$ if $g=\text{argmax}_{\tilde{g}}\hat{z}_{i,\tilde{g}}$ and $\hat{C}^{0}_{i,g}=0$ else) and a prior distribution was chosen for which Expression \eqref{eq: exact_seperate_marglik} has an analytic solution and for which Theorem \ref{thm-nobile_implication} holds. This is the case if the prior on $\tau$ is a Dirichlet distribution and the priors on $\xi_g$ are conjugate and independent. An exact expression of $p(Y|C^0)p(C^0)$ is given in Supplement \hyperref[sec: suppD]{D}.

\paragraph{\textbf{Setting}}

To ensure that Equation \eqref{eq: univ_gauss_mix_supp} can be applied, we needed to sufficiently separate the different components. We did this by choosing an absolute difference of 100 for all mean-component entries and a variance vector of ones for each respective component. We then assessed the validity of Equation \eqref{eq: univ_gauss_mix_supp} empirically: if all allocation vectors generated by the Gibbs sampler were identical, we assumed that there was indeed only one allocation vector with posterior probability not numerically equal to 0.

Let $\text{MVN}_d(\mu_g,\Sigma_g)$ denote the Gaussian law of dimension $d$ with mean $\mu_g$ and covariance matrix $\Sigma_g$. The mixture model is a multivariate Gaussian, \begin{align}\label{eq: gauss_mix_multi_supp}
    Y\in\mathbb{R}^{n\times d},\quad  Y_{1},\dots,Y_{n}|\mu,\Sigma,\tau\stackrel{\text{i.i.d}}\sim \sum^G_{g=1}\tau_g\text{MVN}_d(\mu_g,\Sigma_g).
\end{align} We sampled 1 data set from Equation \eqref{eq: gauss_mix_multi_supp} with identical proportions $\tau_g=\frac{1}{G}$ and covariance matrices $\Sigma_g=I_d$, but different means $\mu_g=100\cdot g\cdot 1_d$, where $1_d$ denotes a vector of ones. We chose the priors\begin{align*}
    \tau&\sim\text{Dirichlet}(G;e_0,\dots,e_0),\quad \tau\in\mathbb{R}^G,\\
    \mu_g|\Sigma_g&\sim\text{MVN}_d(\beta,\Sigma_g/\kappa_0),\quad \mu_1,\dots,\mu_G\in\mathbb{R}^d,\quad g=1,\dots,G,\\
    {\Sigma_g}&{\sim \text{InvWis}(\phi_0,\Lambda)},\quad \Sigma_1,\dots,\Sigma_G\in\mathbb{R}^{d\times d},\quad g=1,\dots,G,\\ e_0&=1,\\\beta_{r,g}&=\frac{1}{n}\sum^n_{i=1}Y_{r,i},\\\kappa_0&=0.00001,\\\phi_0&=d,\\\Lambda&=(1+\phi_0+d)\frac{1}{G}\sum^G_{g=1}\tilde{\Sigma}_g,
\end{align*} where $\tilde{\Sigma}_1,\dots,\tilde{\Sigma}_G$ are estimates for the component-wise covariance matrices, computed via mclust.

Out of all of these values, only the value of $\Lambda$ is non-standard. We suggest this choice of $\Lambda$ since it ensures that the prior mode of the covariance matrices is positioned within a range of plausible values of the covariance matrices. This needed to be done due to the large range of the data, which precludes more standard choices such as the sample covariance matrix, which was several orders of magnitude away from the true value in our case. We also tried to remedy this by scaling the sample covariance by $G^{-d/2}$, as suggested in \citet{FrRa07-bayes_regularization_mixmodels}, but this does in turn result in a value of $\Lambda$ that is very close to the null-matrix, since e.g., for one of the following settings we set $G^{-d/2}=15^{-5/2}$, which is fairly small. 
The hyperparameter that we suggest resulted in a posterior distribution that is fairly concentrated around the true value of each covariance matrix. Gibbs sampling was used, with the sampler being initialised in the allocation vector, using mclust. We tried 2 different settings, each chosen such that it exhibits similarity to a real dataset:

\textbf{Setting 1 (moderate $d$ and $G$)} We simulated from a $d=6$-dimensional model with $G=5$ components and $n=200$ data points. Thus, the dimension of the parameter space is $R=d-1+Gd+Gd(d+1)/2=139$. 

\textbf{Setting 2 (large $d$ and $G$)} We simulated from a $d=5$-dimensional model with $G=15$ components and $n=345$ data points. It is of course possible to use the THAMES to determine the marginal likelihood of the full model, such as in the first setting. However, the number of free parameters is very large, and, as is typical in such high dimensional models, it can be decreased by imposing additional constraints on the covariance matrices of the model. Gibbs sampling is also possible in this case, and a variety of alternative parsimonious geometric models with their associated Gibbs sampling algorithm are presented in \citet{Be_et_al97-BayesianGaussianMixtures}. We choose the following: we impose each covariance matrix to be diagonal with a variance vector $\sigma_{g,1}^2,\dots,\sigma_{g,d}^2$ for each component. We choose independent inverse gamma priors $\text{InvGamma}(\phi_{0},\lambda_j)$ for each of the variances, with $\phi_{0}=2,\lambda_j=2\frac{1}{G}\sum^G_{g=1}\hat{\sigma}_{g,d}^2$ for all $j\in\{1,\dots,d\}$, analogous to the multivariate case.

\subsection*{\textCthree \label{ssec: Cthree}}

The performance of the THAMES was evaluated for the datasets used in \citet{RiGr97-datasets}, namely the $``$acidity$"$, $``$enzyme$"$, and $``$galaxy$"$ dataset. This is because results of other, similar marginal likelihood estimators on the same datasets are available in \citet[Section 2.3.2]{Ce_et_al19-model_selection_mixture_models}. We compared the THAMES to these estimators. 

\paragraph{\textbf{Settings}} All datasets were fitted to a univariate Gaussian mixture model,\begin{align*}
    Y\in\mathbb{R}^{n},\quad  Y_{1},\dots,Y_{n}|\mu,\sigma,\tau\stackrel{\text{i.i.d}}\sim \sum_{g=1}^G\tau_g\mathcal{N}(\mu_g,\sigma^2_g),
\end{align*} with hierarchical priors\begin{align*}
    \tau&\sim\text{Dirichlet}(G;1,\dots,1),\quad \tau\in\mathbb{R}^G\\
    \mu_g&\sim\mathcal{N}(x_0,\lambda),\quad g=1,\dots,G,\\
    {\sigma_g^2}|\zeta&{\sim \text{InvGamma}(2,\zeta)},\quad g=1,\dots,G,\\
    \zeta&\sim\text{Gamma}(0.2,10/\lambda^2),
\end{align*}

where $x_0$ and $\lambda$ are the midpoint and length of the observation interval, respectively.

\paragraph{\textbf{Sampling procedure}}

We used the same sampling procedure as \citet[Section 2.3.2]{Ce_et_al19-model_selection_mixture_models}: Gibbs sampling was conducted to obtain a posterior sample of size $12,000$ with a burn-in of $2,000$, yielding a sample of size $T=10,000$. The values of the galaxy dataset were divided by 100 to be in accordance with \citet[Section 2.3.2]{Ce_et_al19-model_selection_mixture_models}. The bayesmix package \citep{Gr23-bayesmix_package} was used to implement the Gibbs sampler.

\subsection*{\textCfour \label{ssec: Cfour}}

\paragraph{\textbf{Settings}}

The likelihood and priors correspond to those described in Supplement {\color{blue}\hyperref[ssec: Ctwo]{C2}}. 

\paragraph{\textbf{Sampling procedure}}

Gibbs sampling was used, with the sampler being initialised in the allocation vector, using mclust. We used 12,000 iterations with a burn-in of 2,000, resulting in a posterior sample of size $T=10,000$.

\section*{\suppDtext \label{sec: suppD}}

\subsection*{\textDone \label{ssec: Done}}

\begin{proof}[Proof of Theorem \ref{thm-compute_marglik}]By the law of total probability
\begin{align}\label{eq:y_ltp}
    p(Y)&=\sum_{C_1,\dots,C_n}p(Y|C_1,\dots,C_n)p(C_1,\dots,C_n).
\end{align} The prior on the clusters is given by \begin{align}\label{eq:cs_prior}
    p(C_1,\dots,C_n)=\prod^n_{i=1}\prod^G_{g=1}\left(\frac{1}{G}\right)^{C_{i,g}}=G^{-\sum^G_{g=1}\sum^n_{i=1}C_{i,g}},
\end{align} so we only need to solve for $p(Y|C_1,\dots,C_n)$. 
Notice that $Y$ can be written as \begin{align*}
    Y=\left(\prod_{g=1}^G\mu^{C_{1,g}}_g,\dots,\prod_{g=1}^G\mu^{C_{n,g}}_g\right)^\intercal + \varepsilon
\end{align*} where $\varepsilon\sim\text{MVN}(0,I_n)$ is independent from $\mu$. Thus, the marginal distribution of $Y$ given $C_1,\dots,C_n$ is given by \begin{align}\label{eq:marglik_ys_given_cs}
    p(Y|C_1,\dots,C_n)=\text{MVN}_n(Y;0_n,I_n+O).
\end{align}
Plugging \eqref{eq:cs_prior} and \eqref{eq:marglik_ys_given_cs} into \eqref{eq:y_ltp} gives the result.

We are now going to derive the distribution of $\mu$ given $Y$ and $C_1,\dots,C_n$. Let \begin{align*}
    \tilde{\mu}=\left(\prod_{g=1}^G\mu^{C_{1,g}}_g,\dots,\prod_{g=1}^G\mu^{C_{n,g}}_g\right)^\intercal.
\end{align*}

Then, by Bayes' rule \begin{align*}
    &p(\mu|Y,C_1,\dots,C_n)\propto p(\mu)p(C_1,\dots,C_n)p(Y|\mu,C_1,\dots,C_n)\\&\propto p(\mu)p(Y|\mu,C_1,\dots,C_n)\\&=\left(\prod^G_{g=1}\mathcal{N}(\mu_g;0,1)\right)\text{MVN}_n(Y;\tilde{\mu},I_n)\\&\propto\exp\left(-\frac{1}{2}\sum^G_{g=1}\mu_g^2\right)\exp\left(-\frac{1}{2}\left(Y-\tilde{\mu}\right)^\intercal\left(Y-\tilde{\mu}\right)\right)\\&\propto\exp\left(-\frac{1}{2}\sum^G_{g=1}\mu_g^2\right)\exp\left(-\frac{1}{2}\left(\tilde{\mu}^\intercal\tilde{\mu}-2Y^\intercal\tilde{\mu}\right)\right)\\&=\exp\left(-\frac{1}{2}\sum^G_{g=1}\mu_g^2-\frac{1}{2}\sum_{i=1}^{n}\left(\tilde{\mu}_i^2-2\tilde{\mu}_iY_i\right)\right)\\&=\exp\left(-\frac{1}{2}\sum^G_{g=1}\mu_g^2-\frac{1}{2}\sum_{g=1}^{G}\sum_{i:C_{i,g}=1}\left(\mu_g^2-2\mu_gY_i\right)\right).
\end{align*}

We recognise that the density factorises into the densities \begin{align*}
    &p(\mu_g|Y,C_1,\dots,C_n)\propto \exp\left(-\frac{1}{2}\mu_g^2-\frac{1}{2}\sum_{i:C_{i,g}=1}\left(\mu_g^2-2\mu_gY_i\right)\right)\\&=\exp\left(-\frac{1}{2}\left((1+\left(\sum_{i=1}^n C_{i,g}\right))\mu_g^2-2\mu_g \left(\sum_{i=1}^n C_{i,g}Y_i\right)\right)\right)\\&\propto \mathcal{N}\left(\mu_g;\frac{\sum_{i=1}^n C_{i,g}Y_i}{1+\sum_{i=1}^n C_{i,g}},\frac{1}{1+\sum_{i=1}^n C_{i,g}}\right).
\end{align*}

Analogously, \begin{align*}
    &p(C_i|Y,\mu,C_{-i})\propto p(C_i)p(Y|\mu,C_1,\dots,C_n)\\&\propto \exp\left(-\frac{1}{2}\sum_{g=1}^{G}\sum_{i:C_{i,g}=1}\left(\mu_g-Y_i\right)^2\right)\\&\propto \prod^G_{g=1}\exp\left(-\frac{1}{2}C_{i,g}\left(\mu_g-Y_i\right)^2\right)\\&=\prod^G_{g=1}\left(\exp\left(-\frac{1}{2}\left(\mu_g-Y_i\right)^2\right)\right)^{C_{i,g}},
\end{align*}

so the posterior probabilities on $C_i$ given $\mu,Y$ and all other allocation vectors are the normalization of the vectors\begin{align*}
    \left(\exp\left(-\frac{1}{2}\left(\mu_1-Y_i\right)^2\right),\dots,\exp\left(-\frac{1}{2}\left(\mu_G-Y_i\right)^2\right)\right).
\end{align*}
\end{proof}

\subsection*{\textDtwo \label{ssec: Dtwo}}

We know that the marginal likelihood is given by \begin{align*}
    Z(G)=\sum_{C_1,\dots,C_n}p(Y|C_1,\dots,C_n)p(C_1,\dots,C_n)\simeq G! p(Y|C^0)p(C^0)
\end{align*} if $p(Y|C^0)\simeq0$ for all but one allocation matrix $C^0$ (and its permutations). The expression $p(C^0)$ is equal to \citep[see for example][]{Fr06-FiniteMixtureAndMarkovSwitchingModels,Ha_et_al22-QuickMargLikEstimInMixtures}\begin{align*}
    \frac{\Gamma(G)}{\Gamma(G+n)}\prod^G_{g=1}\Gamma\left(1+\sum^n_{i=1}C_{i,g}^0\right).
\end{align*}

First of all, let us note that (denoting $\mu=(\mu_1,\dots,\mu_G),\Sigma=(\Sigma_1,\dots,\Sigma_G)$) \begin{align*}
    p(Y|C^0,\tau,\mu,\Sigma)=p(Y|C^0,\mu,\Sigma),
\end{align*}since the proportions are rendered redundant by the knowledge of the component membership. $p(Y|C^0)$ can thus be further decomposed into \begin{align*}
    p(Y|C^0)&=\int p(Y|C^0,\mu,\Sigma)\pi(\mu,\Sigma)\;d(\mu,\Sigma)\\&=\int\prod_{i=1}^G p(Y_i|C^0,\mu,\Sigma)\pi(\mu,\Sigma)\;d(\mu,\Sigma)\\&=\prod^G_{g=1}\int p_g(Y|\mu_g,\Sigma_g)\pi_g(\mu_g,\Sigma_g)\;d(\mu_g,\Sigma_g)=\prod^G_{g=1}Z_g(G),
\end{align*}where $\pi_g$ denotes the prior on $(\mu_g,\Sigma_g)$ and\begin{align*}
    p_g(Y|\mu_g,\Sigma_g)=\prod_{i:C_{i,g}=1} p(Y_i|C^0,\mu_g,\Sigma_g).
\end{align*} This is because $Y_1,\dots,Y_n$ are independent given $C^0,\mu,\Sigma$ and $(\mu_1,\Sigma_1),\dots,(\mu_G,\Sigma_G)$ are a priori independent.

Because $p_g$ denotes a likelihood on those data points $Y_i$ for which $C_{i,g}=1$ with length $n_g=\sum^n_{i=1}C_{i,g}$, $p_g$ and $\pi_g$ can be interpreted as the prior and likelihood, respectively, of a model with a Gaussian inverse Wishart distribution and the above integrals $Z_g(G)$ can be interpreted as the respective marginal likelihoods of each of these $G$ different models. These marginal likelihoods are tractable due to conjugacy and Bayes rule:\begin{align*}
    &\int p_g(Y|\mu_g,\Sigma_g)\pi_g(\mu_g,\Sigma_g)\;d(\mu_g,\Sigma_g)=\frac{p_g(Y|\mu_g,\Sigma_g)\pi_g(\mu_g,\Sigma_g)}{p_g(\mu_g,\Sigma_g|Y)}\\&\frac{\left(\prod_{i:C_{i,g}=1}\text{MVN}_d(Y_i;\mu_g,\Sigma_g)\right)\text{InvWis}(\Sigma_g;\phi_0,\Lambda)\text{MVN}_d(\mu_g;\beta,\Sigma_g/\kappa_0)}{\text{InvWis}(\Sigma_g;d,\Lambda_{n,g})\text{MVN}_d(\mu_g;\beta_{n,g},\Sigma_g/\kappa_{n,g})}
\end{align*} This holds for any $\mu_g,\Sigma_g$, where $p_g(\mu_g,\Sigma_g|Y)$ denotes the posterior distribution with respect to the $g$'th model. 

The posterior hyperparameters are for example given in \citet{Ge_et_al95-Bayesian_data_analysis}. Plugging in and factoring out $\mu_g,\Sigma_g$ gives the following (rather long) expression, where $\Gamma_d$ denotes the multivariate Gamma function (the posterior hyperparameter $\beta_{n,g}$ ends up being irrelevant since it factors out in the calculation but is still given for completeness): \begin{align*}
    p(Y)&\simeq G!\cdot p(Y|C^0)\cdot p(C^0)\\p(C^0)&=\frac{\Gamma(G)}{\Gamma(G+n)}\prod^G_{g=1}\Gamma\left(1+\sum^n_{i=1}C_{i,g}^0\right)\\p(Y|C^0)&=\prod^G_{g=1}Z_g(G)\\\log(Z_g(G))&=(\phi_0/2)\log(\text{det}(\Lambda))-(\phi_0 d/2)\log(2)-\log(\Gamma_d(\phi_0/2))\\&+(\phi_{n,g}d/2)\log(2)+\Gamma_d(\phi_{n,g}/2)+(-\phi_{n,g}/2)\log(\text{det}(\Lambda_{n,g}))\\&-n_gd\log(\sqrt{2\pi})-2\log(\kappa_{n,g})\\\phi_{n,g}&=\phi_0+n_g\\\Lambda_{n,g}&=\Lambda+\frac{\kappa_0n_g}{\kappa_0+n_g}(\Bar{Y}_g-\beta)(\Bar{Y}_g-\beta)^\intercal+\sum_{i:C_{i,g}=1}(Y_i-\Bar{Y}_g)(Y_i-\Bar{Y}_g)^\intercal\\\beta_{n,g}&=\frac{\kappa_0}{\kappa_0+n_g}\beta+\frac{ n_g}{n_g+\kappa_0}\Bar{Y}_{g}\\\Bar{Y}_{g,r}&=\frac{1}{n_g}\sum_{i:C_{i,g}=1}Y_{i,r}\\\kappa_{n,g}&=\kappa_0+n_g\\n_g&=\sum_{i:C_{i,g}=1}1
\end{align*}

\end{document}